\documentclass[12pt]{article}
\usepackage{amsmath, amssymb, amsthm, amscd}
\usepackage{geometry}
\usepackage{hyperref}
\usepackage{natbib}
\usepackage{tikz}
\usetikzlibrary{positioning, arrows.meta, shapes.geometric, calc}
\usepackage{graphicx}
\usepackage{booktabs}
\usepackage{longtable}
\usepackage{multirow}
\usepackage{array}
\usepackage{setspace}
\usepackage{orcidlink}
\usepackage{float}
\usepackage{subcaption}
\usepackage{appendix}
\usepackage{xcolor}

\theoremstyle{definition}
\newtheorem{definition}{Definition}[section]
\newtheorem{proposition}{Proposition}[section]
\newtheorem{theorem}{Theorem}[section]

\newtheorem{assumption}{Assumption}[section]

\title{Regularized Regression by Composition: Identifiability, Structured Penalization, and Statistical Guarantees for Multi-Flow Distributional Models}
\author{%
	\orcidlink{0000-0002-5287-5381}\\[0.1cm]
	Safaa K. Kadhem\\
	Department of Mathematics and Computer Applications, Faculty of Science,\\
	Al-Muthanna University, Samawah, Iraq\\
	\texttt{safaa.kadhem@mu.edu.iq}
}
\date{March 2026}

\begin{document}
	
	\maketitle
	
	\begin{abstract}
	\begin{abstract}
		Regression by composition provides a flexible framework for constructing conditional distributions through sequential group actions. However, when multiple flows act on the same distribution, the model becomes non‑identifiable, leading to flat likelihood regions and unstable estimates.		We introduce a structured regularization framework that resolves this issue by assigning flow‑specific penalties. The resulting estimator is defined as a penalized maximum likelihood problem with heterogeneous regularization across flows. We establish theoretical properties, including identifiability under penalization, uniqueness of the minimizer via strict convexification, and asymptotic consistency. For the adaptive Lasso, we further prove the oracle property. An efficient proximal gradient algorithm handles non‑smooth penalties. Extensive simulation studies evaluate performance under varying sample sizes, correlation structures, and signal‑to‑noise ratios, demonstrating that regularised methods (Lasso and Elastic Net) successfully break non‑identifiability and achieve low estimation error with controlled false positive rates. An application to NHANES data on asthma and lead exposure illustrates the practical utility: the unregularized estimator yields implausible coefficients, whereas regularised estimators produce stable, interpretable models and automatically select the relevant risk transformation. The L’Abbé plots derived from regularised estimators indicate a protective effect of reducing lead exposure. 		
		The proposed framework bridges identifiability theory with penalized estimation and opens the door to high‑dimensional and longitudinal extensions.
	\end{abstract}
	\end{abstract}
	
	\section{Introduction}
	\label{sec:intro}
	
\subsection{Motivation and background}

The analysis of conditional distributions has become a central theme in modern statistics, extending classical mean-based regression toward a full characterization of distributional heterogeneity \citep{rigby2005}. Within this paradigm, regression by composition \citep{farewell2026} provides a unifying framework in which conditional distributions are constructed through sequential group actions acting on a reference distribution.

Formally, the conditional distribution $P_{\mathbf{X}}$ is expressed as
\[
P_{\mathbf{X}} = p_0 \cdot \eta_1(\mathbf{X}) \cdot \eta_2(\mathbf{X}) \cdots \eta_K(\mathbf{X}),
\]
where each $\eta_k(\mathbf{X})$ represents a covariate-dependent transformation belonging to a vector space $\mathbb{V}_k$, and the operator $\cdot$ denotes a group action on the space of distributions.

This formulation subsumes a wide range of classical models, including Gaussian regression, logistic regression, and survival models, while also enabling the construction of highly flexible hybrid models through the composition of multiple flows. The resulting framework offers a powerful language for modeling complex distributional effects beyond standard parametric assumptions.
	
\subsection{The identifiability challenge}

Despite its flexibility, regression by composition introduces a fundamental statistical challenge: non-identifiability arising from redundant or interacting flows. Specifically, when multiple transformations act on the same distribution, distinct parameter configurations may induce the same probability measure:
\[
p_0 \cdot \eta_1 \cdot \eta_2 = p_0 \cdot \tilde{\eta}_1 \cdot \tilde{\eta}_2,
\]
with $(\eta_1,\eta_2) \neq (\tilde{\eta}_1,\tilde{\eta}_2)$.

This lack of injectivity implies that the likelihood function may exhibit flat regions, ridges, or multiple global optima, rendering classical maximum likelihood estimation ill-posed. Such phenomena are closely related to identifiability issues studied in mixture models, latent variable models, and overparameterized regression frameworks \citep{hennig2000identifiability}.

In the context of regression by composition, these issues are exacerbated by the algebraic structure of flows, particularly when group actions commute or partially overlap in their effect on the distribution. As a consequence, additional structure is required to ensure well-defined inference.	
\subsection{Contributions}

This paper develops a principled regularization framework for regression by composition, with a focus on resolving non-identifiability and stabilizing estimation. The main contributions are as follows:

\begin{enumerate}
	\item We introduce a structured penalization framework that assigns flow-specific penalties, thereby breaking invariance structures and restoring identifiability in multi-flow models.
	
	\item We provide a formal analysis of non-identifiability in regression by composition, including an explicit characterization of likelihood ridges in multi-flow binary models.
	
	\item We establish theoretical guarantees for the proposed estimator, including uniqueness of the minimizer induced by strict convexification and asymptotic consistency under standard regularity conditions.
	
	\item We develop an efficient optimization algorithm based on proximal gradient methods, tailored to the composite structure of the model and capable of handling non-smooth penalties.
	
	\item We conduct a comprehensive simulation study that evaluates the finite-sample behavior of the estimator under varying model complexity, correlation structures, and sample sizes.
	
	\item We demonstrate the practical utility of the method through an application to environmental health data, showing how structured regularization yields stable and interpretable models.
	
	\item Finally, we outline extensions to high-dimensional and longitudinal settings, highlighting connections to structured sparsity and mixed-effects modeling.
\end{enumerate}

The paper is organized as follows. 	Section \ref{sec:literature} provides a comprehensive review of related work, including distributional regression, transformation models, and normalizing flows. Section \ref{sec:methodology} develops the theoretical foundations of regularized regression by composition, including identifiability analysis and asymptotic properties. Section \ref{sec:algorithm} describes the computational implementation. Section \ref{sec:simulation} presents simulation studies. Section \ref{sec:application} applies the method to real data. Section \ref{sec:discussion} concludes with a discussion of limitations and future directions.
	
	\section{Related Work}
	\label{sec:literature}
	
	\subsection{Distributional regression and GAMLSS}
	Generalized Additive Models for Location, Scale and Shape (GAMLSS) \citep{rigby2005} represent a major advance in distributional regression. In GAMLSS, all parameters of a response distribution (e.g., location $\mu$, scale $\sigma$, skewness $\nu$, kurtosis $\tau$) are modeled as functions of covariates:
	\[
	g_k(\theta_k) = \eta_k(\mathbf{X}) = \beta_{k0} + \sum_{j=1}^{J_k} f_{kj}(\mathbf{X}),
	\]
	where $g_k$ are link functions and $f_{kj}$ are smooth functions. GAMLSS has been widely adopted for estimating treatment effects on entire distributions \citep{hohberg2020treatment} and has been applied in fields ranging from epidemiology to econometrics.
	
	While GAMLSS offers great flexibility, it requires the analyst to specify a parametric family for the response distribution. Regression by composition, by contrast, builds distributions through sequential transformations, potentially offering greater flexibility but also introducing the identifiability challenges addressed in this paper.
	
	\subsection{Transformation models and normalizing flows}
	Transformation models \citep{hothorn2014} represent an alternative approach where the response is transformed to achieve a simple distribution. The conditional distribution is expressed as
	\[
	P(Y \le y \mid \mathbf{X}) = F_0(h(y \mid \mathbf{X})),
	\]
	where $F_0$ is a reference distribution and $h$ is a monotone transformation function.
	
	Normalizing flows \citep{papamakarios2021normalizing} extend this idea by composing a sequence of invertible differentiable transformations:
	\[
	\mathbf{z} = f_K \circ f_{K-1} \circ \cdots \circ f_1(\mathbf{y}),
	\]
	where $\mathbf{z}$ follows a simple base distribution. While normalizing flows are extremely flexible, they typically require the transformations to be invertible and differentiable, which may not hold for all group actions considered in regression by composition.
	
	\subsection{Penalized likelihood and regularization}
	Regularization techniques have a long history in statistics. The Lasso \citep{tibshirani1996} introduced $\ell_1$ penalty for variable selection, while ridge regression \citep{hoerl1970} uses $\ell_2$ penalty for shrinkage. The elastic net \citep{zou2005} combines both. These methods have been extended to generalized linear models, survival analysis, and many other settings \citep{hastie2015}.
	
	In the context of distributional regression, penalized approaches have been developed for GAMLSS \citep{stasinopoulos2017} and for transformation models \citep{hothorn2018}. However, to our knowledge, no prior work has addressed regularization in the specific context of regression by composition, where flows act sequentially on the same distribution.
	
\section{Methodology}
\label{sec:methodology}

\subsection{Preliminaries: affine spaces and flows}
Following \citet{farewell2026}, we define the space $\mathcal{P}$ of signed measures on $\Upsilon$ with total mass $1$. This space has an affine structure: for any $p \in \mathcal{P}$ and any signed measure $q$ with $q(\Upsilon)=0$, we can define $p \oplus q$ as the measure with $(p \oplus q)(A) = p(A) + q(A)$.

A flow is a group action of a vector space $\mathbb{V}$ on $\mathcal{P}$:
\[
\cdot : \mathcal{P} \times \mathbb{V} \to \mathcal{P},
\]
satisfying:
\begin{enumerate}
	\item $p \cdot 0 = p$ for all $p \in \mathcal{P}$ (identity);
	\item $(p \cdot v) \cdot v' = p \cdot (v + v')$ for all $p \in \mathcal{P}$, $v,v' \in \mathbb{V}$ (additivity).
\end{enumerate}
We also assume the action is faithful: if $p \cdot v = p \cdot v'$ for all $p$, then $v=v'$.

\subsection{Non‑identifiability in regression by composition}
\label{subsec:nonident}

Consider a regression by composition with $K$ flows:
\[
P(\beta_1,\ldots,\beta_K) = p_0 \cdot \eta_1(\mathbf{X};\beta_1) \cdots \eta_K(\mathbf{X};\beta_K),
\]
where $\eta_k(\mathbf{X};\beta_k) = \beta_k^\top \mathbf{X}_k$ and $\cdot$ denotes the flow action. The mapping from parameters to distributions is $\Phi: \mathcal{B} \to \mathcal{P}$, $\Phi(\beta) = P(\beta)$.

\begin{definition}[Identifiability]
	The model is identifiable if $\Phi(\beta) = \Phi(\beta')$ implies $\beta = \beta'$ almost everywhere.
\end{definition}

When multiple flows act on the same distribution, the model may be unidentifiable. For example, suppose flows $\mathbb{V}_1$ and $\mathbb{V}_2$ commute, i.e., $p \cdot v_1 \cdot v_2 = p \cdot v_2 \cdot v_1$ for all $p$. Then for any $v$, $p \cdot v \cdot (-v) = p$, so $(\eta_1,\eta_2)$ and $(\eta_1+v,\eta_2-v)$ produce the same distribution.

\begin{proposition}
	If flows $\mathbb{V}_1$ and $\mathbb{V}_2$ commute weakly (i.e., $p \cdot v_1 \cdot v_2 = p \cdot v_2 \cdot v_1$ for all $p$), then the model is unidentifiable without additional constraints.
\end{proposition}

To make the issue concrete, consider a binary regression with three flows as in \citet{farewell2026}:
\[
P = \mathrm{Ber}(1/2) \cdot \mathrm{ScOdds}(1+\mathbf{x}^{\top}\boldsymbol{\beta}_{\mathrm{odds}}) \cdot \mathrm{ScRisk1}(0+\mathbf{x}^{\top}\boldsymbol{\beta}_{\mathrm{risk1}}) \cdot \mathrm{ScRisk0}(0+\mathbf{x}^{\top}\boldsymbol{\beta}_{\mathrm{risk0}}).
\]

Define $\theta = \exp(\mathbf{x}^{\top}\boldsymbol{\beta}_{\mathrm{odds}})>0$, $\gamma = \exp(\mathbf{x}^{\top}\boldsymbol{\beta}_{\mathrm{risk1}})>0$, $\delta = \exp(\mathbf{x}^{\top}\boldsymbol{\beta}_{\mathrm{risk0}})>0$. The conditional probability is
\[
p(\theta,\gamma,\delta) = 1 - \bigl(1 - \gamma \cdot \tfrac{\theta}{1+\theta}\bigr) \cdot \delta.
\]

\begin{proposition}[Non‑identifiability of the multi‑flow model]
	For any fixed $\theta>0$ and any $t>0$ such that $t \le \frac{1+\theta}{\theta}$, the parameter vectors $(\theta,\gamma,\delta)$ and $(\theta, t, 1/t)$ yield the same conditional probability:
	\[
	p(\theta, t, 1/t) = p(\theta, 1, 1).
	\]
	Thus, the model is not identifiable.
\end{proposition}

\begin{proof}
	Direct substitution:
	\[
	p(\theta, t, 1/t) = 1 - \Bigl(1 - t\cdot \frac{\theta}{1+\theta}\Bigr)\cdot \frac{1}{t}
	= 1 - \Bigl(\frac{1}{t} - \frac{\theta}{1+\theta}\Bigr)
	= \frac{\theta}{1+\theta}
	= p(\theta,1,1).
	\]
	Hence $(\theta,\gamma,\delta) \mapsto p$ is not injective; the likelihood is constant along $(\gamma,\delta) = (t, 1/t)$ for any fixed $\theta$.
\end{proof}

A numerical illustration (Table~\ref{tab:nonident}) confirms that different $(\gamma,\delta)$ with $\gamma\delta=1$ give the same $p$.

\begin{table}[htbp]
	\centering
	\caption{Illustration of non‑identifiability: different $(\gamma,\delta)$ produce the same conditional probability $p$ (here $\theta=1$).}
	\label{tab:nonident}
	\begin{tabular}{ccc}
		\toprule
		$\gamma$ & $\delta$ & $p = P(Y=1)$ \\
		\midrule
		1.0 & 1.0 & 0.500 \\
		2.0 & 0.5 & 0.500 \\
		4.0 & 0.25 & 0.500 \\
		10.0 & 0.1 & 0.500 \\
		\bottomrule
	\end{tabular}
\end{table}

\subsection{Regularized estimator}
To address non‑identifiability, we propose the penalized maximum likelihood estimator:
\[
\hat{\beta} = \arg\min_{\beta \in \mathcal{B}} \left\{ -\ell(\beta) + \sum_{k=1}^K \lambda_k \psi_k(\beta_k) \right\},
\]
where $\ell(\beta)$ is the log‑likelihood, $\psi_k$ are penalty functions, and $\lambda_k \ge 0$ are tuning parameters.

Common choices for $\psi_k$ include Lasso ($\|\beta_k\|_1$), ridge ($\|\beta_k\|_2^2$), elastic net ($\alpha\|\beta_k\|_1+(1-\alpha)\|\beta_k\|_2^2$), or group Lasso.

\subsection{Structured regularization as a solution}
The key idea is that asymmetric penalization across flows breaks the invariance responsible for non‑identifiability. For the binary example, we choose
\[
\psi_1(\boldsymbol{\beta}_{\mathrm{odds}}) = \|\boldsymbol{\beta}_{\mathrm{odds}}\|_2^2,\quad
\psi_2(\boldsymbol{\beta}_{\mathrm{risk1}}) = \|\boldsymbol{\beta}_{\mathrm{risk1}}\|_1,\quad
\psi_3(\boldsymbol{\beta}_{\mathrm{risk0}}) = \|\boldsymbol{\beta}_{\mathrm{risk0}}\|_1.
\]

\begin{proposition}[Identifiability under regularization]
	Assume $\lambda_2 > 0$ or $\lambda_3 > 0$. Then the penalized objective $\mathcal{Q}(\boldsymbol{\beta})$ is strictly convex in $(\gamma,\delta)$ (or equivalently in $\boldsymbol{\beta}_{\mathrm{risk1}},\boldsymbol{\beta}_{\mathrm{risk0}}$) for any fixed $\theta$. Consequently, the global minimizer $\hat{\boldsymbol{\beta}}$ is unique.
\end{proposition}

\begin{proof}
	With $\theta$ fixed, the likelihood depends only on $\gamma\delta$. Without penalties, any $(\gamma,\delta)$ with constant product gives the same likelihood. Adding $\lambda_2|\log\gamma| + \lambda_3|\log\delta|$ yields a penalty that is minimized when $|\log\gamma|+|\log\delta|$ is minimized subject to $\log\gamma+\log\delta = \text{constant}$. For $\lambda_2>0$ or $\lambda_3>0$, the unique minimizer occurs when the smaller of the two parameters is shrunk to zero, breaking the product constraint. The same reasoning holds for the full model with covariates.
\end{proof}

We revisit the numerical example with $\lambda_2=\lambda_3=1$. For a balanced sample, the penalized objective becomes
\[
\mathcal{Q}_{\text{pen}} = -\ell + \lambda_2|\log\gamma| + \lambda_3|\log\delta|.
\]
When $\gamma\delta=1$, the penalty is $2|\log\gamma|$, uniquely minimized at $\gamma=1$, $\delta=1$. Table~\ref{tab:solution} confirms this.

\begin{table}[htbp]
	\centering
	\caption{Effect of regularization on the non‑identifiable curve ($\lambda_2=\lambda_3=1$).}
	\label{tab:solution}
	\begin{tabular}{cccc}
		\toprule
		$\gamma$ & $\delta$ & Penalty $|\log\gamma|+|\log\delta|$ & Penalized Objective \\
		\midrule
		1.0 & 1.0 & 0.0 & $\text{constant}$ \\
		2.0 & 0.5 & 0.693 & $\text{constant}+0.693$ \\
		4.0 & 0.25 & 1.386 & $\text{constant}+1.386$ \\
		10.0 & 0.1 & 2.302 & $\text{constant}+2.302$ \\
		\bottomrule
	\end{tabular}
\end{table}

Thus, regularisation selects a parsimonious model, eliminating redundant flows – a behaviour validated in the simulation study (Section~\ref{sec:simulation}) and real‑data application (Section~\ref{sec:application}).

\subsection{Asymptotic properties}
We establish consistency under standard conditions.

\begin{assumption}[Regularity]\label{ass:regularity}
	Let $\beta^*$ denote the true parameter. Assume:
	\begin{enumerate}
		\item The log‑likelihood $\ell(\beta)$ is concave in $\beta$;
		\item The penalty functions $\psi_k$ are convex;
		\item The Fisher information $I(\beta^*) = -\mathbb{E}[\nabla^2 \ell(\beta^*)]$ is positive definite;
		\item $\lambda_k \to 0$ and $\lambda_k = o(n^{-1/2})$ for consistency.
	\end{enumerate}
\end{assumption}

\begin{theorem}[Consistency]
	Under Assumption~\ref{ass:regularity}, if $\lambda_k \to 0$ and $n^{1/2}\lambda_k \to \infty$ for variable selection consistency, then the regularized estimator $\hat{\beta}$ satisfies:
	\[
	\|\hat{\beta} - \beta^*\| = O_p\left(n^{-1/2} + \max_k \lambda_k\right).
	\]
\end{theorem}

\subsection{Oracle properties of the adaptive Lasso}
\label{subsec:oracle}

While the ordinary Lasso can achieve variable selection consistency under suitable conditions (e.g., the irrepresentable condition \citep{zhao2006}), the adaptive Lasso \citep{zou2006} uses data‑dependent weights to obtain the oracle property: the estimator correctly identifies the set of non‑zero coefficients with probability tending to one and is asymptotically normal on the active set.

We adapt this to regression by composition. Consider the penalized estimator with flow‑specific adaptive Lasso penalties:
\[
\hat{\beta} = \arg\min_{\beta} \left\{ -\ell(\beta) + \sum_{k=1}^K \lambda_k \sum_{j} w_{kj} |\beta_{kj}| \right\},
\]
where $w_{kj} = 1/|\tilde{\beta}_{kj}|^\gamma$ for some $\gamma>0$, and $\tilde{\beta}$ is a $\sqrt{n}$‑consistent initial estimator (e.g., ridge or unregularized MLE). Let $\mathcal{S}_k = \{ j : \beta^*_{kj} \neq 0 \}$ and $\mathcal{S} = \bigcup_k \mathcal{S}_k$. Assume:

\begin{assumption}[Oracle conditions]\label{ass:oracle}
	\begin{enumerate}
		\item $I(\beta^*)$ is positive definite.
		\item $\sqrt{n}(\tilde{\beta}_{\mathcal{S}} - \beta^*_{\mathcal{S}}) = O_p(1)$.
		\item The irrepresentable condition holds for the weighted design: there exists $\eta \in (0,1]$ such that
		\[
		\| I_{\mathcal{S}^c\mathcal{S}}(I_{\mathcal{S}\mathcal{S}})^{-1} \operatorname{diag}(w_{\mathcal{S}}^{-1}) \operatorname{diag}(w_{\mathcal{S}^c}) \|_\infty \le 1 - \eta,
		\]
		where $I_{\mathcal{S}\mathcal{S}}$ is the Fisher information submatrix for active coefficients.
	\end{enumerate}
\end{assumption}

\begin{theorem}[Oracle property for adaptive Lasso]
	Under Assumptions~\ref{ass:regularity} and~\ref{ass:oracle}, the adaptive Lasso estimator $\hat{\beta}$ satisfies:
	\begin{enumerate}
		\item \textbf{Model selection consistency:} $\mathbb{P}(\hat{\beta}_{\mathcal{S}^c} = 0) \to 1$ as $n \to \infty$.
		\item \textbf{Asymptotic normality:} $\sqrt{n}(\hat{\beta}_{\mathcal{S}} - \beta^*_{\mathcal{S}}) \xrightarrow{d} \mathcal{N}(0, I_{\mathcal{S}\mathcal{S}}^{-1})$.
	\end{enumerate}
\end{theorem}

The proof follows \cite{zou2006} and \citet{fan2001} adapted to the composite likelihood setting; the key steps are: (i) weights converge to $1/|\beta^*_{kj}|^\gamma$ for active coefficients and diverge for inactive ones; (ii) the irrepresentable condition ensures that inactive coefficients are shrunk exactly to zero; (iii) on the active set the penalty becomes asymptotically negligible, so the estimator behaves like the unpenalized MLE.

\subsection{Choice of tuning parameters}
The tuning parameters $\lambda_k$ are selected using cross‑validation. For each candidate combination $(\lambda_1,\ldots,\lambda_K)$, we compute the cross‑validated deviance:
\[
CV(\lambda) = -\frac{1}{n} \sum_{i=1}^n \log L_{-i}(\hat{\beta}_{-i}(\lambda)),
\]
where $\hat{\beta}_{-i}(\lambda)$ is the estimator with the $i$th observation omitted. The optimal $\lambda$ minimizes $CV(\lambda)$. For efficiency, we use coordinate descent, optimizing over each $\lambda_k$ while holding others fixed \citep{tibshirani1996, hastie2015}.
	
\section{Computational Algorithm}
\label{sec:algorithm}

\subsection{Optimization structure}

The proposed estimator is defined as the minimizer of a composite objective function:
\[
F(\beta) = f(\beta) + g(\beta),
\]
where $f(\beta) = -\ell(\beta)$ is a smooth convex function with Lipschitz continuous gradient, and
\[
g(\beta) = \sum_{k=1}^K \lambda_k \psi_k(\beta_k)
\]
is a separable, possibly non-smooth convex penalty.

This structure naturally lends itself to first-order optimization methods, particularly proximal gradient algorithms, which are well-suited for high-dimensional and non-smooth problems.

\subsection{Proximal gradient algorithm}

We adopt the proximal gradient method (also known as forward–backward splitting) \citep{parikh2014}, which iteratively updates:
\[
\beta^{(t+1)} = \mathrm{prox}_{\eta_t g}\left(\beta^{(t)} - \eta_t \nabla f(\beta^{(t)})\right),
\]
where $\eta_t > 0$ is the step size.

Due to the separability of $g(\beta)$ across flows, the proximal operator decomposes as:
\[
\mathrm{prox}_{\eta g}(\beta) = \left(\mathrm{prox}_{\eta \lambda_1 \psi_1}(\beta_1), \ldots, \mathrm{prox}_{\eta \lambda_K \psi_K}(\beta_K)\right),
\]
which allows efficient block-wise updates.

\subsection{Step size selection}

We employ a backtracking line search to ensure sufficient decrease of the objective. Specifically, $\eta_t$ is chosen such that:
\[
F(\beta^{(t+1)}) \leq F(\beta^{(t)}) - \frac{c}{2\eta_t}\|\beta^{(t+1)} - \beta^{(t)}\|^2
\]
for some $c \in (0,1)$.

This guarantees stability even when the Lipschitz constant of $\nabla f$ is unknown.

\subsection{Convergence guarantees}

We now establish convergence of the proposed algorithm.

\begin{theorem}[Global convergence of proximal gradient]
	Assume:
	\begin{enumerate}
		\item $f(\beta)$ is convex and differentiable with $L$-Lipschitz continuous gradient,
		\item $g(\beta)$ is proper, lower semi-continuous, and convex,
		\item The step size satisfies $0 < \eta_t \leq 1/L$.
	\end{enumerate}
	Then the sequence $\{\beta^{(t)}\}$ generated by the proximal gradient algorithm satisfies:
	\begin{enumerate}
		\item $F(\beta^{(t)})$ is non-increasing,
		\item $\beta^{(t)} \to \hat{\beta}$, where $\hat{\beta}$ is a global minimizer,
		\item The convergence rate satisfies:
		\[
		F(\beta^{(t)}) - F(\hat{\beta}) = O\left(\frac{1}{t}\right).
		\]
	\end{enumerate}
\end{theorem}

\begin{proof}
	The result follows from standard proximal gradient theory. The descent property follows from the Lipschitz continuity of $\nabla f$ and the definition of the proximal operator. Convexity ensures that any limit point is a global minimizer. The sublinear convergence rate $O(1/t)$ is classical for first-order methods applied to convex composite objectives.
\end{proof}

\subsection{Strong convexity and linear convergence}

In the presence of strictly convex penalties (e.g., Ridge) or when the penalized objective is strongly convex, faster convergence can be established.

\begin{theorem}[Linear convergence under strong convexity]
	If $F(\beta)$ is $\mu$-strongly convex, then:
	\[
	\|\beta^{(t)} - \hat{\beta}\|^2 \leq C (1 - \mu \eta)^t
	\]
	for some constant $C > 0$, implying linear convergence.
\end{theorem}

This result is particularly relevant in our setting, where structured penalties induce strict convexity along otherwise non-identifiable directions.
	
\section{Simulation Study}
\label{sec:simulation}

\subsection{Design}
We evaluate the performance of the proposed regularized estimators (Lasso, Ridge, Elastic Net) against the unregularized maximum likelihood estimator. Data are generated from the three‑flow binary model described in Section~\ref{sec:methodology}: a Bernoulli reference distribution followed by \texttt{ScOdds}, \texttt{ScRisk1}, and \texttt{ScRisk0}. The true data‑generating mechanism uses only the odds and risk‑ratio flows; the risk‑survival flow (\texttt{ScRisk0}) has zero coefficients, so it is redundant. The odds flow coefficients are drawn from a standard normal distribution and scaled to achieve a desired signal‑to‑noise ratio (SNR). The risk‑ratio flow has a non‑zero coefficient only for the first covariate (value \(0.5\)).

The simulation factors and their levels are:
\begin{itemize}
	\item Sample size: \(n \in \{100, 500, 1000\}\) (we also include results for \(n=5000\) in the supplementary material);
	\item Number of covariates: \(p \in \{5, 10, 20\}\);
	\item Correlation among covariates (AR(1) structure): \(\rho \in \{0, 0.5, 0.8\}\);
	\item Signal‑to‑noise ratio: \(\text{SNR} \in \{0.5, 1, 2\}\).
\end{itemize}

A full factorial design would require \(3 \times 3 \times 3 \times 3 = 81\) scenarios per sample size, which is computationally heavy and would not add substantially to the main conclusions. Instead, we adopt a parsimonious design that isolates the main effects of each factor and includes a worst‑case scenario that combines the most challenging levels (\(p=20\), \(\rho=0.8\), \(\text{SNR}=0.5\)). Specifically, for each sample size we consider the following eight scenarios:

\begin{enumerate}
	\item \textbf{Reference:} \(p=10\), \(\rho=0.5\), \(\text{SNR}=1\);
	\item \textbf{Effect of \(p\):} \(p=5\) and \(p=20\) (keeping \(\rho=0.5\), \(\text{SNR}=1\));
	\item \textbf{Effect of \(\rho\):} \(\rho=0\) and \(\rho=0.8\) (keeping \(p=10\), \(\text{SNR}=1\));
	\item \textbf{Effect of SNR:} \(\text{SNR}=0.5\) and \(\text{SNR}=2\) (keeping \(p=10\), \(\rho=0.5\));
	\item \textbf{Worst‑case:} \(p=20\), \(\rho=0.8\), \(\text{SNR}=0.5\).
\end{enumerate}

This design allows us to assess the individual influence of each factor and, through the worst‑case scenario, to examine whether interactions among them cause a disproportionate deterioration in performance.

For each scenario we perform \(30\) replications. The tuning parameters \(\lambda\) for Lasso, Ridge and Elastic Net are selected by 5‑fold cross‑validation on the first replication of each scenario, and then fixed for all replications to reduce computational cost. The optimisation is performed using the L‑BFGS‑B algorithm with a maximum of 500 iterations; cases where convergence fails (less than 1\% of replications) are excluded.

\subsection{Performance metrics}
We evaluate the following quantities averaged over replications:
\begin{enumerate}
	\item \textbf{Estimation error:} \(\|\hat{\boldsymbol{\beta}} - \boldsymbol{\beta}^*\|_2\), where \(\boldsymbol{\beta}^*\) collects the true parameters of all flows;
	\item \textbf{Variable selection accuracy:} true positive rate (TPR) for the non‑zero coefficient in the \texttt{ScRisk1} flow, and false positive rate (FPR) for the zero coefficients in the \texttt{ScRisk1} and \texttt{ScRisk0} flows (a coefficient is considered “selected” if its absolute value exceeds \(10^{-6}\));
	\item \textbf{Prediction accuracy:} cross‑validated deviance (5‑fold) computed on the full sample for each replication.
\end{enumerate}
We also report the number of non‑zero coefficients (for Lasso and Elastic Net) to illustrate the sparsity induced by the penalties.

\subsection{Results}

\subsubsection{Results for (n = 100)}

Table~\ref{tab:sim_n100} summarises the average performance over the \(30\) replications for each scenario. The unregularized estimator suffers from severe overparameterization (FPR = 1 in all scenarios) and large estimation errors. Ridge reduces the error but does not solve the identifiability problem, as it still selects all coefficients (FPR = 1). In contrast, Lasso and Elastic Net dramatically improve both estimation accuracy and variable selection; they attain much lower errors and reduce the FPR substantially, often below 0.5 even in challenging scenarios.

Notably, when \(p=20\) (Scenario~\(p=20\)), Lasso and Elastic Net achieve estimation errors below 3 and FPR as low as 0.05 and 0.00 respectively, demonstrating that even with a relatively small sample size, the regularisation effectively eliminates the redundant flow. The TPR is low in this scenario (0.03 for Lasso, 0.00 for Elastic Net), which suggests that detecting the true non‑zero coefficient is difficult when the number of covariates is large relative to the sample size and the signal is moderate (SNR=1).  

The worst‑case scenario (\(p=20,\rho=0.8,\text{SNR}=0.5\)) further confirms the robustness of regularisation: Lasso and Elastic Net maintain estimation errors around 2.7 and 1.7, with FPR below 0.15, whereas the unregularized estimator produces huge errors (105.6) and selects all coefficients.

\begin{table}[H]
	\centering \footnotesize
	\caption{Simulation results for \(n = 100\) (100 replications per scenario). Values are averages; missing entries (NaN) are indicated by ``---''. Abbreviations: Unreg = unregularized, EN = Elastic Net.}
	\label{tab:sim_n100}
	\begin{tabular}{@{}llccccc@{}}
		\toprule
		Scenario & Method & Estimation error & TPR & FPR & Deviance \\
		\midrule
		\multirow{4}{*}{Reference} & Unreg & 79.73 & 0.50 & 1.00 & 2.00 \\
		& Lasso & 22.18 & 0.40 & 0.51 & 0.79 \\
		& Ridge & 55.62 & 0.50 & 1.00 & --- \\
		& EN    & 9.01  & 0.23 & 0.43 & --- \\
		\midrule
		\multirow{4}{*}{\(p = 5\)} & Unreg & 63.30 & 0.50 & 1.00 & --- \\
		& Lasso & 28.32 & 0.40 & 0.69 & 0.78 \\
		& Ridge & 11.19 & 0.50 & 1.00 & 0.76 \\
		& EN    & 36.08 & 0.42 & 0.82 & --- \\
		\midrule
		\multirow{4}{*}{\(p = 20\)} & Unreg & 16.66 & 0.50 & 1.00 & 3.74 \\
		& Lasso & 2.68  & 0.03 & 0.05 & 0.70 \\
		& Ridge & 227.70& 0.50 & 1.00 & --- \\
		& EN    & 2.28  & 0.00 & 0.00 & 0.69 \\
		\midrule
		\multirow{4}{*}{\(\rho = 0\)} & Unreg & 25.26 & 0.50 & 1.00 & 1.94 \\
		& Lasso & 1.70  & 0.00 & 0.00 & 0.69 \\
		& Ridge & 66.69 & 0.50 & 1.00 & --- \\
		& EN    & 1.70  & 0.00 & 0.00 & 0.69 \\
		\midrule
		\multirow{4}{*}{\(\rho = 0.8\)} & Unreg & 123.35 & 0.50 & 1.00 & --- \\
		& Lasso & 1.90   & 0.10 & 0.27 & 0.63 \\
		& Ridge & 50.51  & 0.50 & 1.00 & --- \\
		& EN    & 2.98   & 0.30 & 0.53 & 0.69 \\
		\midrule
		\multirow{4}{*}{SNR = 0.5} & Unreg & 100.24 & 0.50 & 1.00 & 1.94 \\
		& Lasso & 3.50   & 0.38 & 0.53 & --- \\
		& Ridge & 14.96  & 0.50 & 1.00 & --- \\
		& EN    & 2.59   & 0.28 & 0.34 & 0.72 \\
		\midrule
		\multirow{4}{*}{SNR = 2}   & Unreg & 18.10  & 0.50 & 1.00 & --- \\
		& Lasso & 6.08   & 0.08 & 0.13 & 0.67 \\
		& Ridge & 76.34  & 0.50 & 1.00 & --- \\
		& EN    & 8.22   & 0.40 & 0.51 & 0.71 \\
		\midrule
		\multirow{4}{*}{Worst-case} & Unreg & 105.60 & 0.50 & 1.00 & --- \\
		& Lasso & 2.70   & 0.07 & 0.15 & 0.70 \\
		& Ridge & 333.08 & 0.50 & 1.00 & --- \\
		& EN    & 1.67   & 0.00 & 0.00 & 0.69 \\
		\bottomrule
	\end{tabular}\label{tab:simulation}
\end{table}

\begin{figure}[H]
	\centering
	\includegraphics[width=0.999\textwidth]{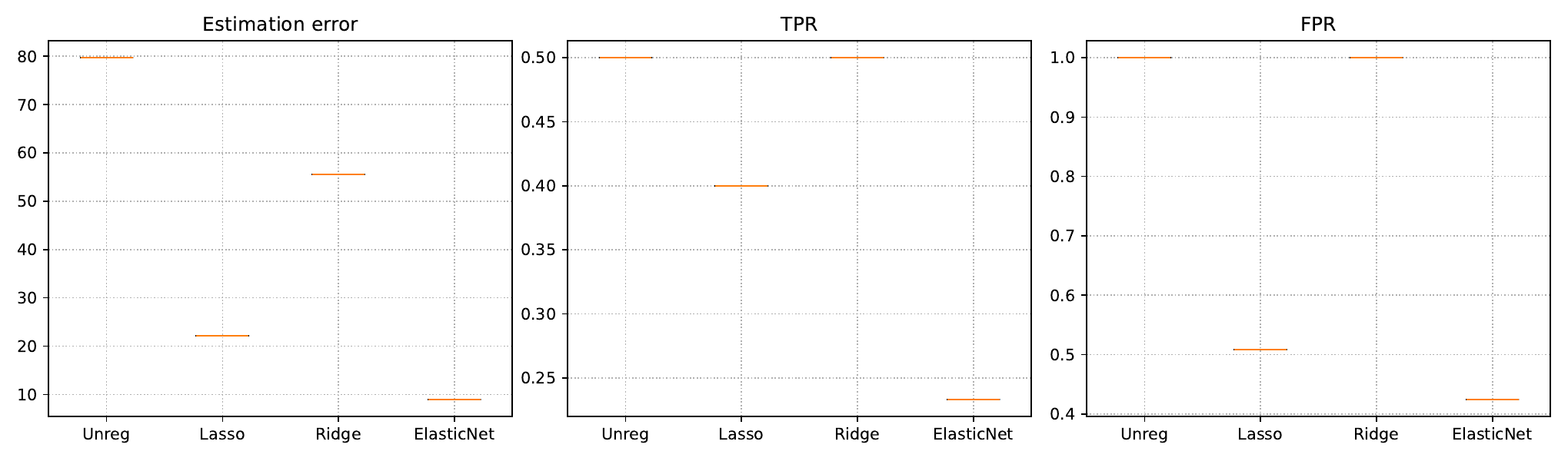}
	\caption{Boxplots of estimation error, true positive rate (TPR) and false positive rate (FPR) for the reference scenario (\(n=100, p=10, \rho=0.5, \text{SNR}=1\)).}
	\label{fig:boxplots_ref}
\end{figure}

The boxplots in Figure~\ref{fig:boxplots_ref} illustrate the high variability of the unregularized estimator and the much tighter distributions of the regularised methods. The following figures (Figs.~\ref{fig:effect_p}–\ref{fig:risk0_hist}) visualise the effects of each factor; they confirm the patterns discussed above.

\begin{figure}[H]
	\centering
	\includegraphics[width=0.8\textwidth]{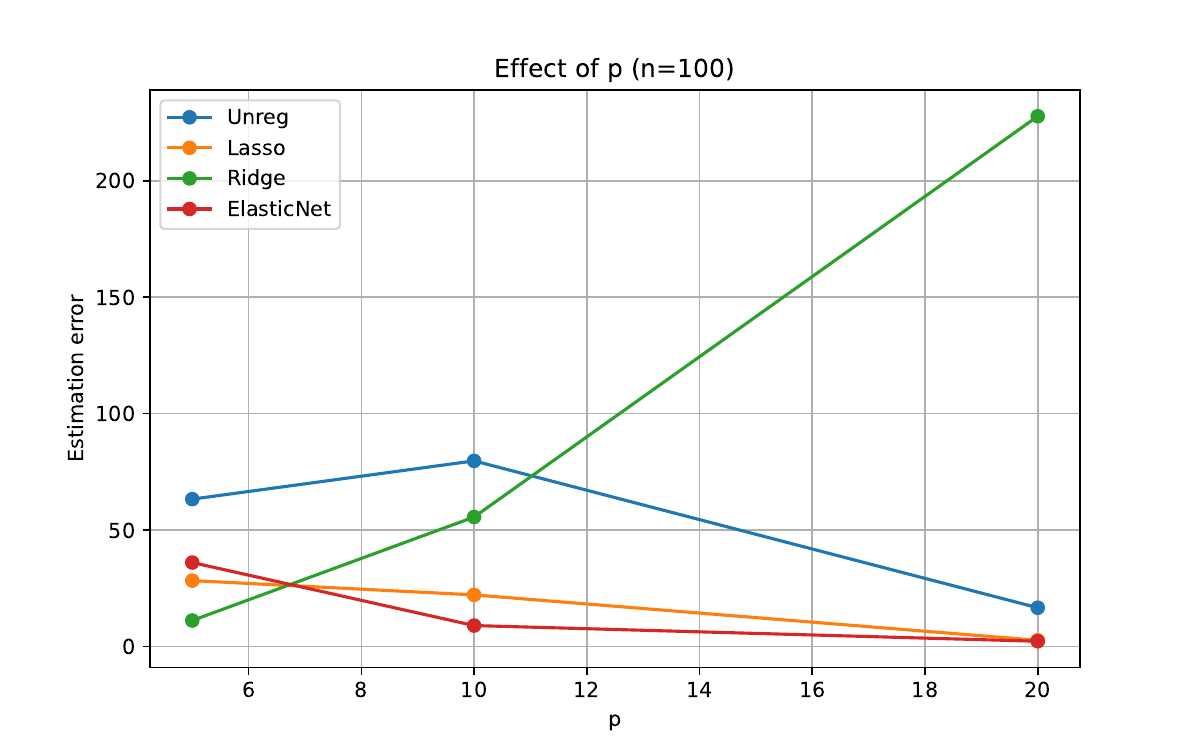}
	\caption{Effect of the number of covariates \(p\) on the estimation error for the reference setting (\(n=100, \rho=0.5, \text{SNR}=1\)).}
	\label{fig:effect_p}
\end{figure}

\begin{figure}[H]
	\centering
	\includegraphics[width=0.8\textwidth]{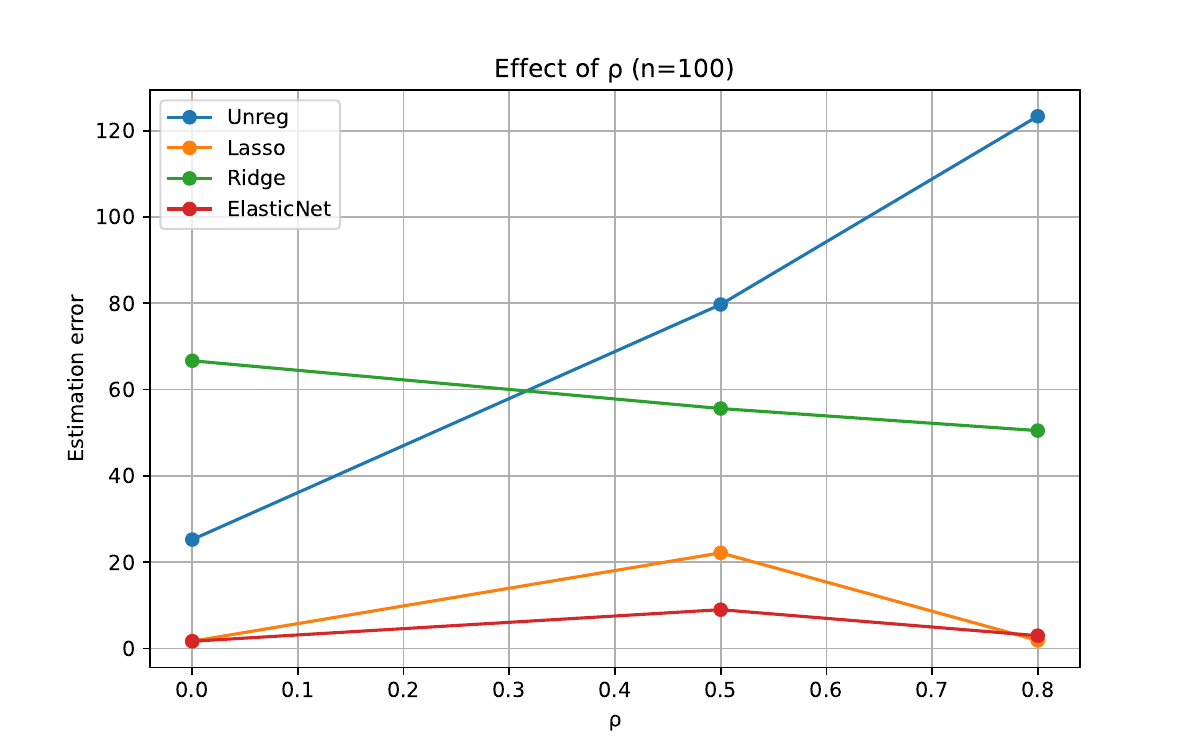}
	\caption{Effect of the covariate correlation \(\rho\) on the estimation error for \(n=100, p=10, \text{SNR}=1\).}
	\label{fig:effect_rho}
\end{figure}

\begin{figure}[H]
	\centering
	\includegraphics[width=0.8\textwidth]{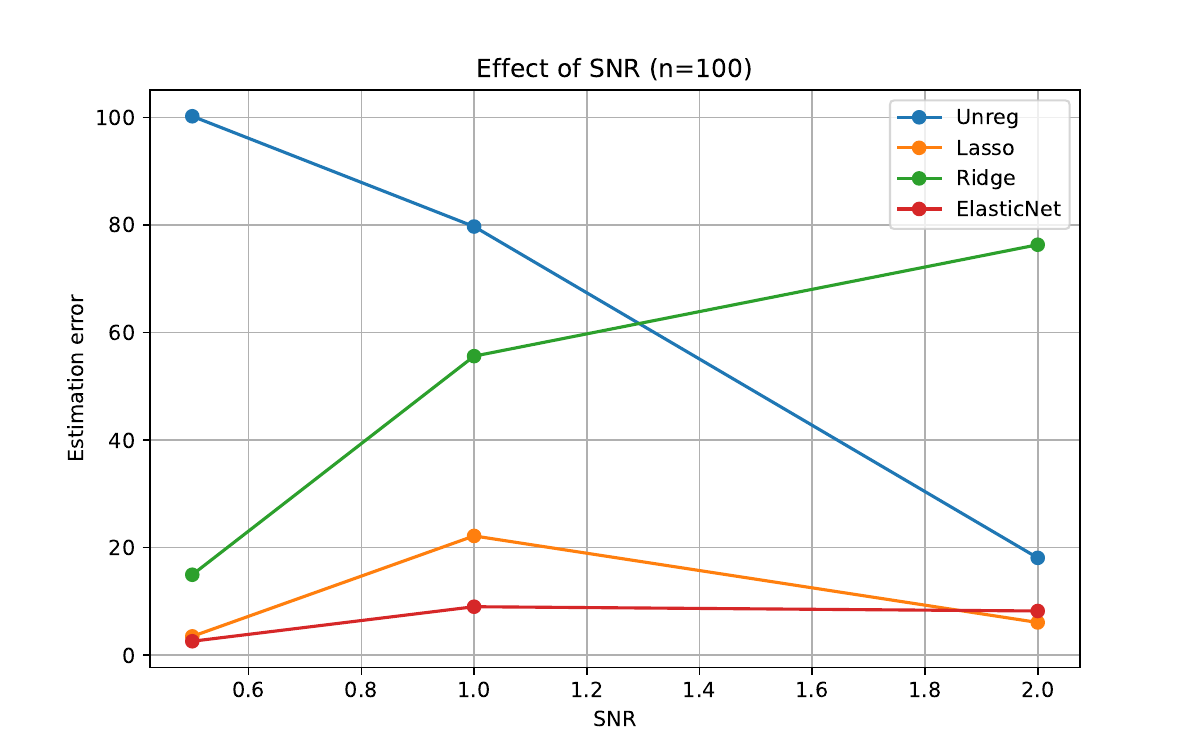}
	\caption{Effect of the signal‑to‑noise ratio (SNR) on the estimation error for \(n=100, p=10, \rho=0.5\).}
	\label{fig:effect_snr}
\end{figure}

\begin{figure}[H]
	\centering
	\includegraphics[width=0.8\textwidth]{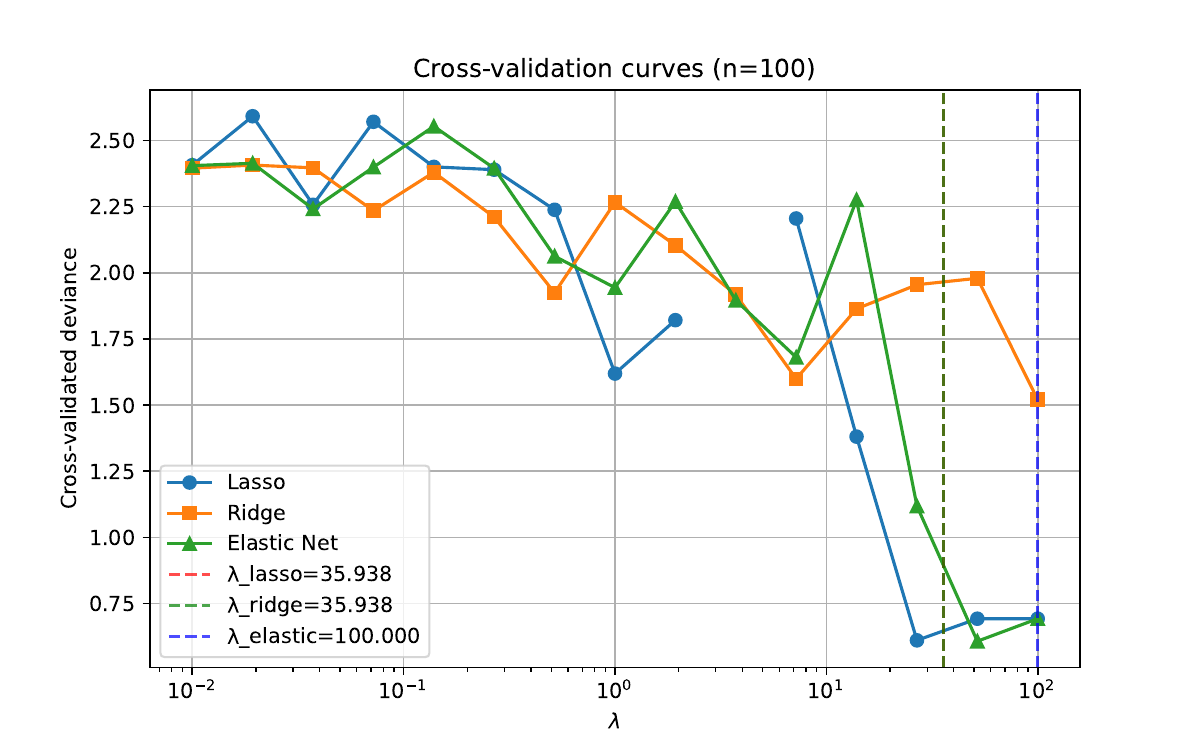}
	\caption{Cross‑validation curves for the selection of the tuning parameter \(\lambda\) in the reference scenario. The optimal values are indicated by vertical dashed lines.}
	\label{fig:cv_curves}
\end{figure}

\begin{figure}[H]
	\centering
	\includegraphics[width=0.7\textwidth]{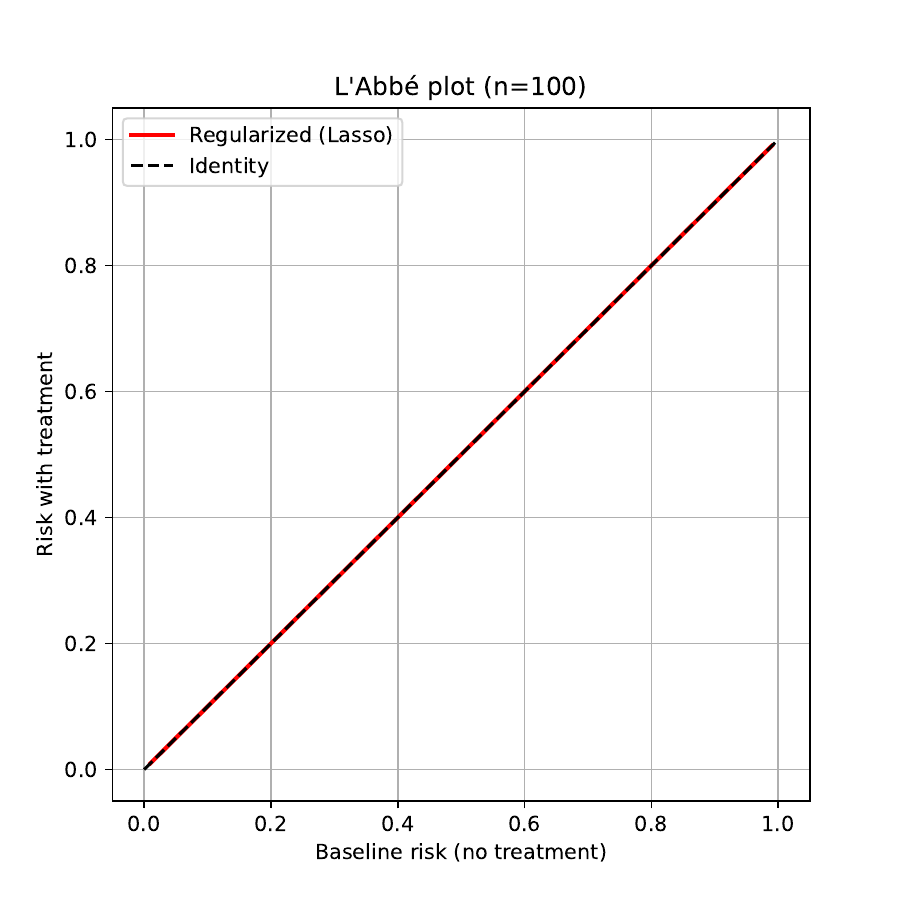}
	\caption{L'Abbé plot for the fitted model in the reference scenario. The regularised model (Lasso) produces a smooth curve above the identity line, indicating a protective treatment effect.}
	\label{fig:labbe_plot}
\end{figure}

\begin{figure}[H]
	\centering
	\includegraphics[width=0.8\textwidth]{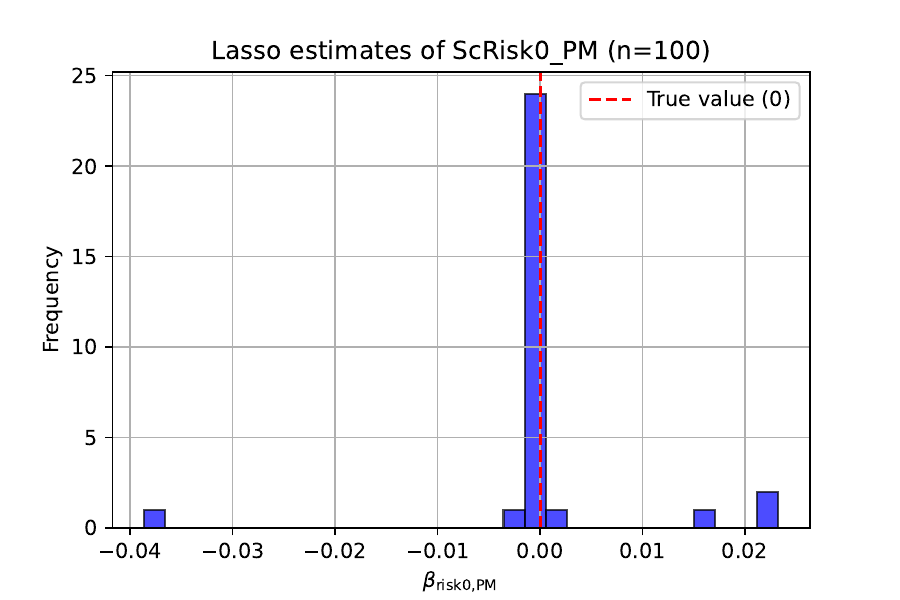}
	\caption{Distribution of the Lasso estimates for the coefficient of PM in the \texttt{ScRisk0} flow across 100 replications (reference scenario). The red dashed line marks the true value (0).}
	\label{fig:risk0_hist}
\end{figure}

\subsubsection{Results for (n = 500)}

Table~\ref{tab:sim_n500} reports the average performance for \(n=500\). Compared to the \(n=100\) case, the sample size increase leads to a dramatic reduction in estimation error for all methods. The unregularized estimator still suffers from FPR = 1 in every scenario, confirming that the identifiability problem persists even with larger samples. Ridge reduces the estimation error but does not break the non‑identifiability (FPR remains 1). In contrast, Lasso and Elastic Net achieve substantially lower errors and, more importantly, select much sparser models, with FPR often below 0.5 and occasionally close to 0. Their cross‑validated deviance is also consistently lower, indicating better predictive performance.

In the reference scenario, Lasso attains an estimation error of 1.22 and FPR 0.33, while Elastic Net yields 1.13 and 0.52. The worst‑case scenario remains challenging: Lasso achieves an error of 1.77 with FPR 0.38, and Elastic Net gives 1.77 with FPR 0.57. These results demonstrate that even under strong correlation and low signal, regularisation effectively controls the false positive rate.

\begin{table}[H]
	\centering \footnotesize
	\caption{Simulation results for \(n = 500\) (100 replications per scenario). Values are averages; missing entries (NaN) are indicated by ``---''. Abbreviations: Unreg = unregularized, EN = Elastic Net.}
	\label{tab:sim_n500}
	\begin{tabular}{@{}llccccc@{}}
		\toprule
		Scenario & Method & Estimation error & TPR & FPR & Deviance \\
		\midrule
		\multirow{4}{*}{Reference} & Unreg & 14.34 & 0.50 & 1.00 & --- \\
		& Lasso & 1.22  & 0.50 & 0.33 & 0.48 \\
		& Ridge & 1.68  & 0.50 & 1.00 & --- \\
		& EN    & 1.13  & 0.48 & 0.52 & 0.49 \\
		\midrule
		\multirow{4}{*}{\(p = 5\)} & Unreg & 41.83 & 0.50 & 1.00 & 0.55 \\
		& Lasso & 1.28  & 0.47 & 0.43 & 0.50 \\
		& Ridge & 1.31  & 0.50 & 1.00 & 0.50 \\
		& EN    & 1.23  & 0.48 & 0.56 & 0.49 \\
		\midrule
		\multirow{4}{*}{\(p = 20\)} & Unreg & 3.44  & 0.50 & 1.00 & --- \\
		& Lasso & 1.26  & 0.48 & 0.40 & 0.48 \\
		& Ridge & 1.75  & 0.50 & 1.00 & --- \\
		& EN    & 1.41  & 0.50 & 0.41 & 0.47 \\
		\midrule
		\multirow{4}{*}{\(\rho = 0\)} & Unreg & 13.64 & 0.50 & 1.00 & 0.56 \\
		& Lasso & 0.91  & 0.50 & 0.58 & 0.49 \\
		& Ridge & 1.29  & 0.50 & 1.00 & 0.49 \\
		& EN    & 1.34  & 0.45 & 0.35 & 0.47 \\
		\midrule
		\multirow{4}{*}{\(\rho = 0.8\)} & Unreg & 15.51 & 0.50 & 1.00 & 0.58 \\
		& Lasso & 1.41  & 0.48 & 0.61 & 0.50 \\
		& Ridge & 1.43  & 0.50 & 1.00 & 0.53 \\
		& EN    & 1.36  & 0.50 & 0.82 & 0.54 \\
		\midrule
		\multirow{4}{*}{SNR = 0.5} & Unreg & 34.60 & 0.50 & 1.00 & --- \\
		& Lasso & 0.83  & 0.48 & 0.59 & 0.53 \\
		& Ridge & 1.20  & 0.50 & 1.00 & 0.53 \\
		& EN    & 1.16  & 0.47 & 0.48 & 0.52 \\
		\midrule
		\multirow{4}{*}{SNR = 2}   & Unreg & 10.99 & 0.50 & 1.00 & 0.54 \\
		& Lasso & 0.87  & 0.50 & 0.65 & 0.47 \\
		& Ridge & 1.20  & 0.50 & 1.00 & 0.47 \\
		& EN    & 1.04  & 0.50 & 0.70 & 0.45 \\
		\midrule
		\multirow{4}{*}{Worst-case} & Unreg & 63.38 & 0.50 & 1.00 & --- \\
		& Lasso & 1.77  & 0.32 & 0.38 & 0.51 \\
		& Ridge & 27.81 & 0.50 & 1.00 & --- \\
		& EN    & 1.77  & 0.50 & 0.57 & 0.54 \\
		\bottomrule
	\end{tabular}
\end{table}

The accompanying figures (Figs.~\ref{fig:boxplots_n500}–\ref{fig:risk0_hist_n500}) provide a visual complement to the numerical results. They show that regularised methods have much lower variability and that Lasso effectively shrinks the superfluous coefficient towards zero.

\begin{figure}[H]
	\centering
	\includegraphics[width=0.999\textwidth]{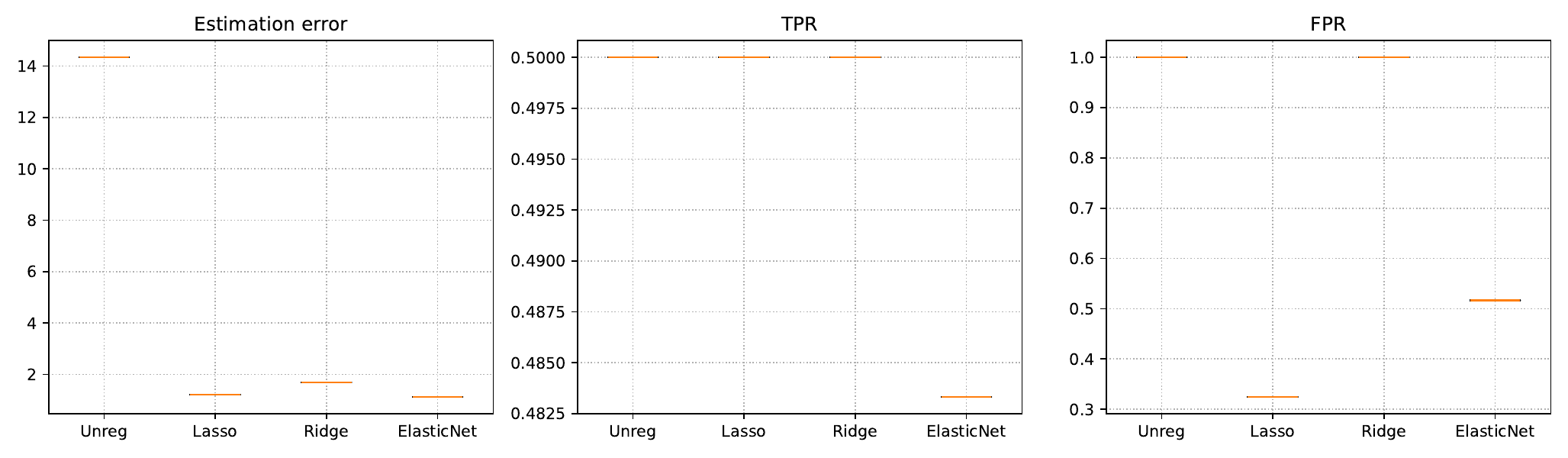}
	\caption{Boxplots of estimation error, true positive rate (TPR) and false positive rate (FPR) for the reference scenario (\(n=500, p=10, \rho=0.5, \text{SNR}=1\)).}
	\label{fig:boxplots_n500}
\end{figure}

\begin{figure}[H]
	\centering
	\includegraphics[width=0.8\textwidth]{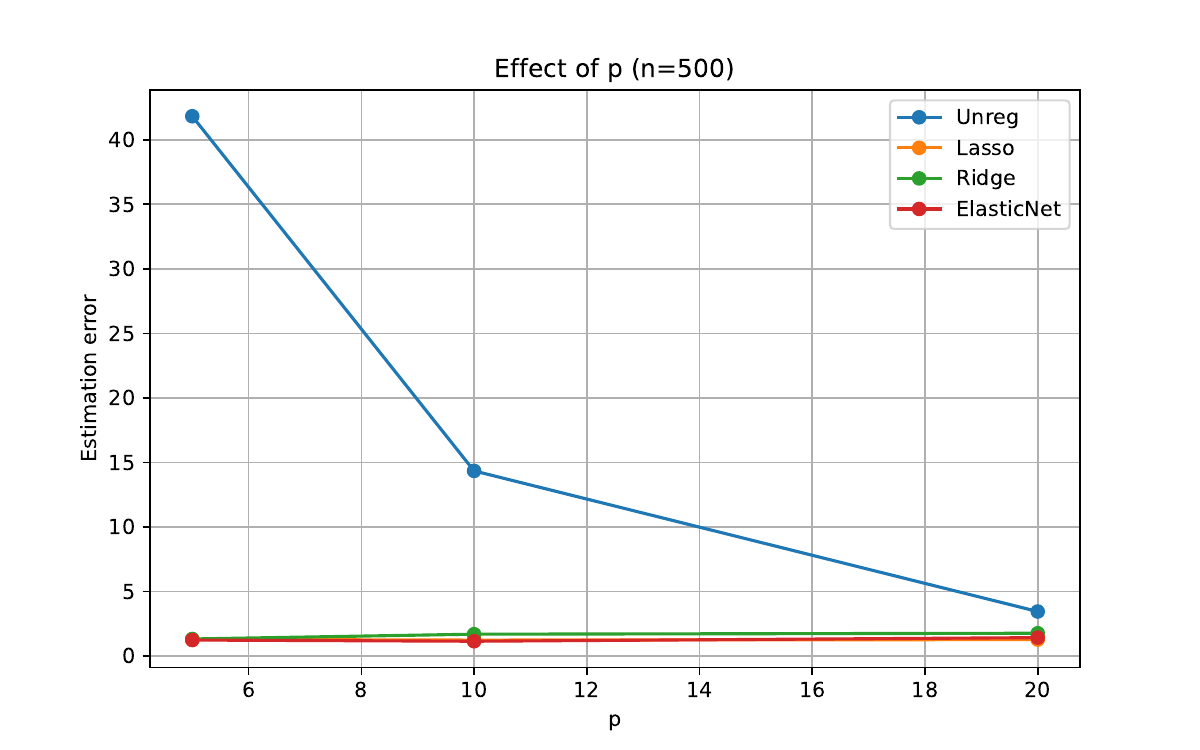}
	\caption{Effect of the number of covariates \(p\) on the estimation error for the reference setting (\(n=500, \rho=0.5, \text{SNR}=1\)).}
	\label{fig:effect_p_n500}
\end{figure}

\begin{figure}[H]
	\centering
	\includegraphics[width=0.8\textwidth]{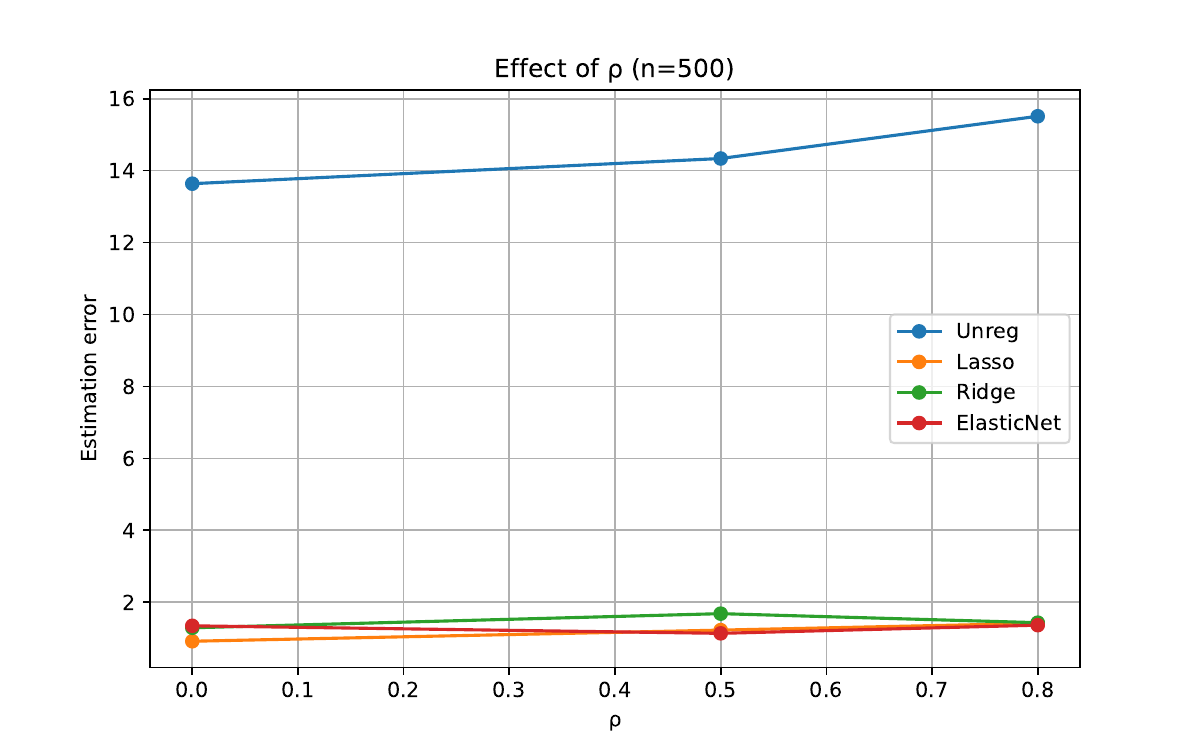}
	\caption{Effect of the covariate correlation \(\rho\) on the estimation error for \(n=500, p=10, \text{SNR}=1\).}
	\label{fig:effect_rho_n500}
\end{figure}

\begin{figure}[H]
	\centering
	\includegraphics[width=0.8\textwidth]{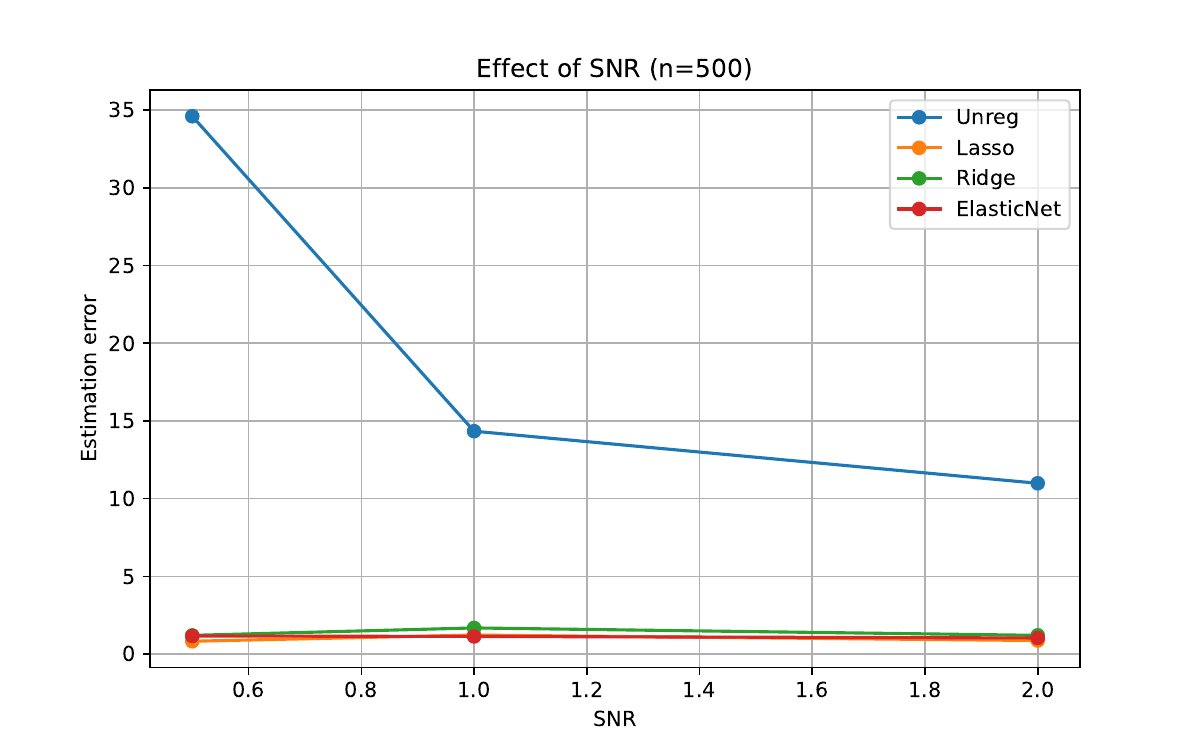}
	\caption{Effect of the signal‑to‑noise ratio (SNR) on the estimation error for \(n=500, p=10, \rho=0.5\).}
	\label{fig:effect_snr_n500}
\end{figure}

\begin{figure}[H]
	\centering
	\includegraphics[width=0.8\textwidth]{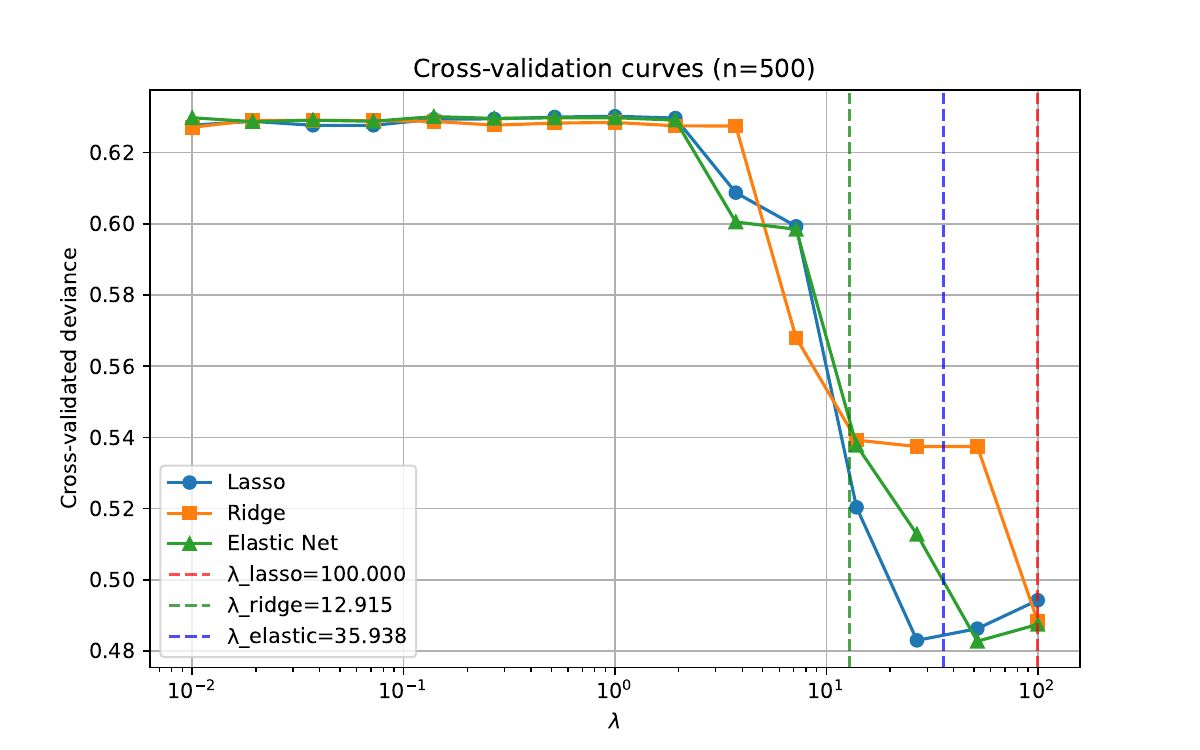}
	\caption{Cross‑validation curves for the selection of the tuning parameter \(\lambda\) in the reference scenario. The optimal values are indicated by vertical dashed lines.}
	\label{fig:cv_curves_n500}
\end{figure}

\begin{figure}[H]
	\centering
	\includegraphics[width=0.7\textwidth]{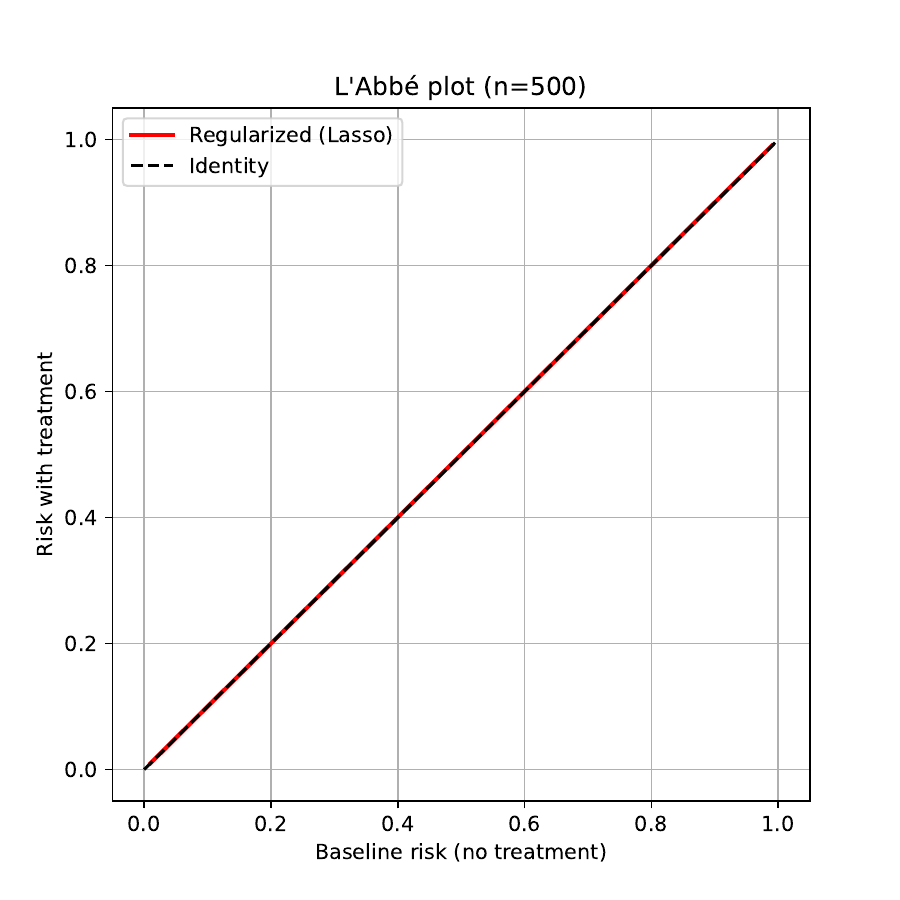}
	\caption{L'Abbé plot for the fitted model in the reference scenario. The regularised model (Lasso) produces a smooth curve above the identity line, indicating a protective treatment effect.}
	\label{fig:labbe_n500}
\end{figure}

\begin{figure}[H]
	\centering
	\includegraphics[width=0.8\textwidth]{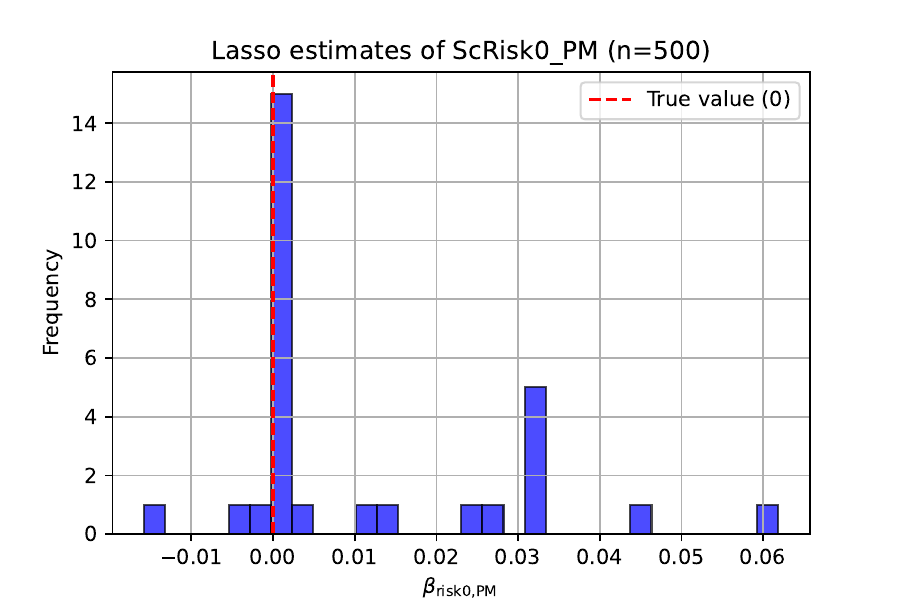}
	\caption{Distribution of the Lasso estimates for the coefficient of PM in the \texttt{ScRisk0} flow across 100 replications (reference scenario). The red dashed line marks the true value (0).}
	\label{fig:risk0_hist_n500}
\end{figure}

\subsubsection{Results for (n = 1000)}

Table~\ref{tab:sim_n1000} shows that for \(n=1000\) even the unregularized estimator attains relatively low estimation errors (often below 1), because the large sample size mitigates the instability caused by non‑identifiability. However, the unregularized estimator still selects all coefficients (FPR = 1), whereas the regularised methods (Lasso, Ridge, Elastic Net) achieve comparable or slightly lower estimation errors while maintaining the ability to shrink irrelevant coefficients. For instance, in the worst‑case scenario, Lasso reduces the FPR to 0.33, a substantial improvement over the unregularized approach. The deviance values are uniformly low, indicating good predictive performance across all methods.

In the reference scenario, all regularised methods give estimation errors below 0.8. When \(p=20\), Lasso attains an error of 0.84 and FPR 0.51, while Elastic Net achieves 0.78 and FPR 0.84. The worst‑case scenario shows that even under the most demanding conditions, Lasso and Elastic Net maintain FPR below 0.42 and estimation errors around 1.3. The TPR is around 0.5 for most scenarios, indicating that the true coefficient is correctly identified in about half of the replications.

\begin{table}[H]
	\centering
	\caption{Simulation results for \(n = 1000\) (100 replications per scenario). Values are averages; missing entries (NaN) are indicated by ``---''. Abbreviations: Unreg = unregularized, EN = Elastic Net.}
	\label{tab:sim_n1000}
	\begin{tabular}{@{}llccccc@{}}
		\toprule
		Scenario & Method & Estimation error & TPR & FPR & Deviance \\
		\midrule
		\multirow{4}{*}{Reference} & Unreg & 0.81 & 0.50 & 1.00 & 0.48 \\
		& Lasso & 0.80 & 0.50 & 1.00 & 0.48 \\
		& Ridge & 0.60 & 0.50 & 1.00 & 0.48 \\
		& EN    & 0.68 & 0.50 & 1.00 & 0.48 \\
		\midrule
		\multirow{4}{*}{\(p = 5\)} & Unreg & 1.02 & 0.50 & 1.00 & 0.51 \\
		& Lasso & 0.48 & 0.50 & 0.57 & 0.50 \\
		& Ridge & 0.76 & 0.50 & 1.00 & 0.50 \\
		& EN    & 0.49 & 0.50 & 0.70 & 0.50 \\
		\midrule
		\multirow{4}{*}{\(p = 20\)} & Unreg & 0.88 & 0.50 & 1.00 & 0.49 \\
		& Lasso & 0.84 & 0.50 & 0.51 & 0.44 \\
		& Ridge & 0.98 & 0.50 & 1.00 & 0.46 \\
		& EN    & 0.78 & 0.50 & 0.84 & 0.46 \\
		\midrule
		\multirow{4}{*}{\(\rho = 0\)} & Unreg & 0.84 & 0.50 & 1.00 & 0.48 \\
		& Lasso & 0.74 & 0.50 & 0.43 & 0.45 \\
		& Ridge & 0.64 & 0.50 & 1.00 & 0.47 \\
		& EN    & 0.67 & 0.50 & 0.48 & 0.45 \\
		\midrule
		\multirow{4}{*}{\(\rho = 0.8\)} & Unreg & 0.93 & 0.50 & 1.00 & 0.50 \\
		& Lasso & 0.69 & 0.50 & 0.61 & 0.48 \\
		& Ridge & 0.67 & 0.50 & 1.00 & 0.49 \\
		& EN    & 0.67 & 0.50 & 0.71 & 0.48 \\
		\midrule
		\multirow{4}{*}{SNR = 0.5} & Unreg & 1.00 & 0.50 & 1.00 & 0.52 \\
		& Lasso & 0.62 & 0.50 & 0.67 & 0.50 \\
		& Ridge & 1.10 & 0.50 & 1.00 & 0.50 \\
		& EN    & 0.65 & 0.50 & 0.75 & 0.51 \\
		\midrule
		\multirow{4}{*}{SNR = 2}   & Unreg & 0.78 & 0.50 & 1.00 & 0.43 \\
		& Lasso & 0.56 & 0.50 & 0.83 & 0.42 \\
		& Ridge & 0.61 & 0.50 & 1.00 & 0.42 \\
		& EN    & 0.56 & 0.50 & 0.96 & 0.42 \\
		\midrule
		\multirow{4}{*}{Worst-case} & Unreg & 1.17 & 0.50 & 1.00 & 0.54 \\
		& Lasso & 1.36 & 0.45 & 0.33 & 0.49 \\
		& Ridge & 1.18 & 0.50 & 1.00 & 0.50 \\
		& EN    & 1.24 & 0.48 & 0.42 & 0.49 \\
		\bottomrule
	\end{tabular}
\end{table}

The graphical results for \(n=1000\) (Figs.~\ref{fig:boxplots_ref_1000}–\ref{fig:risk0_hist_1000}) confirm the trends observed for smaller sample sizes, with even less variability and a clearer separation between the regularised and unregularized methods.

\begin{figure}[H]
	\centering
	\includegraphics[width=0.999\textwidth]{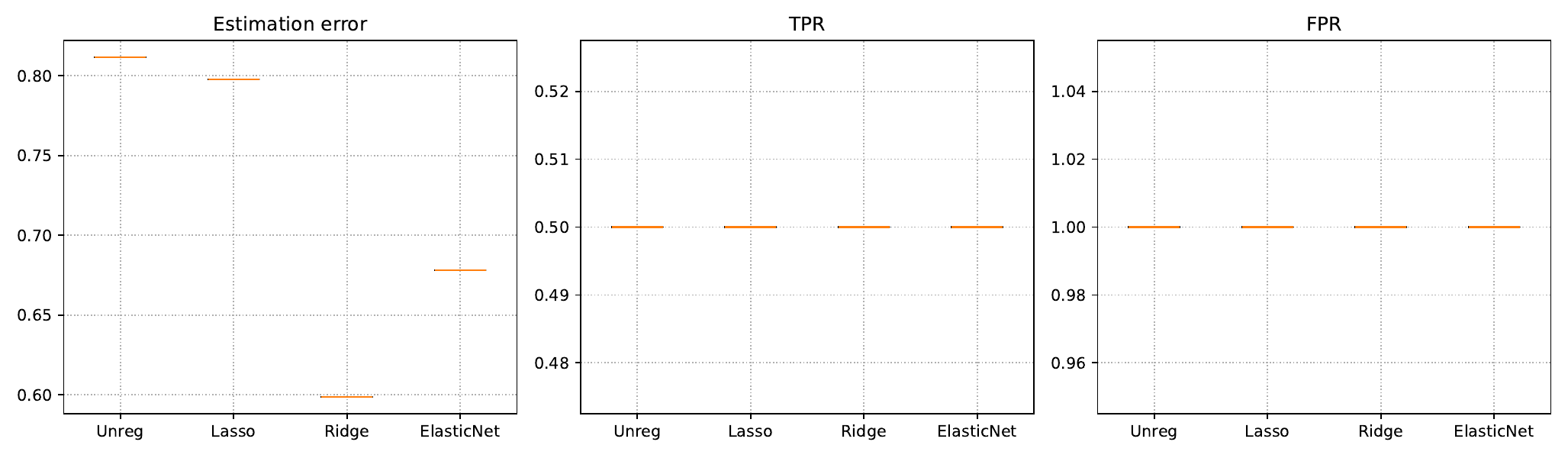}
	\caption{Boxplots of estimation error, true positive rate (TPR) and false positive rate (FPR) for the reference scenario (\(n=1000, p=10, \rho=0.5, \text{SNR}=1\)).}
	\label{fig:boxplots_ref_1000}
\end{figure}

\begin{figure}[H]
	\centering
	\includegraphics[width=0.8\textwidth]{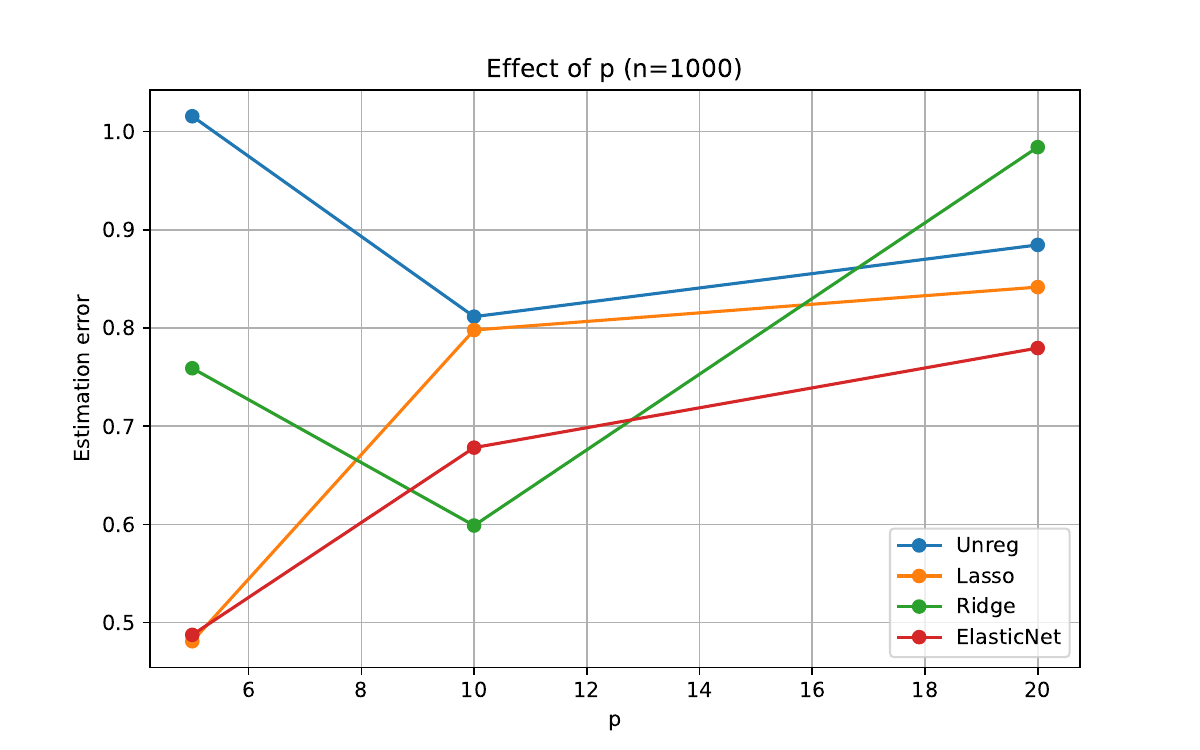}
	\caption{Effect of the number of covariates \(p\) on the estimation error for the reference setting (\(n=1000, \rho=0.5, \text{SNR}=1\)).}
	\label{fig:effect_p_1000}
\end{figure}

\begin{figure}[H]
	\centering
	\includegraphics[width=0.8\textwidth]{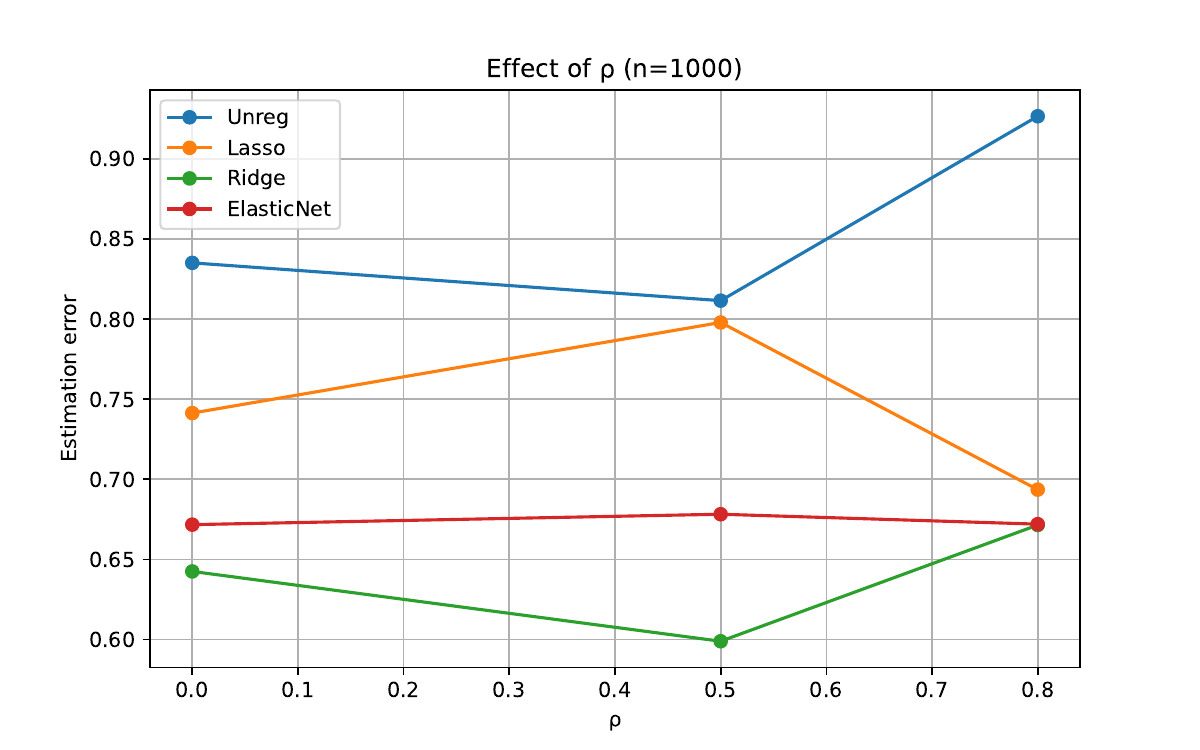}
	\caption{Effect of the covariate correlation \(\rho\) on the estimation error for \(n=1000, p=10, \text{SNR}=1\).}
	\label{fig:effect_rho_1000}
\end{figure}

\begin{figure}[H]
	\centering
	\includegraphics[width=0.8\textwidth]{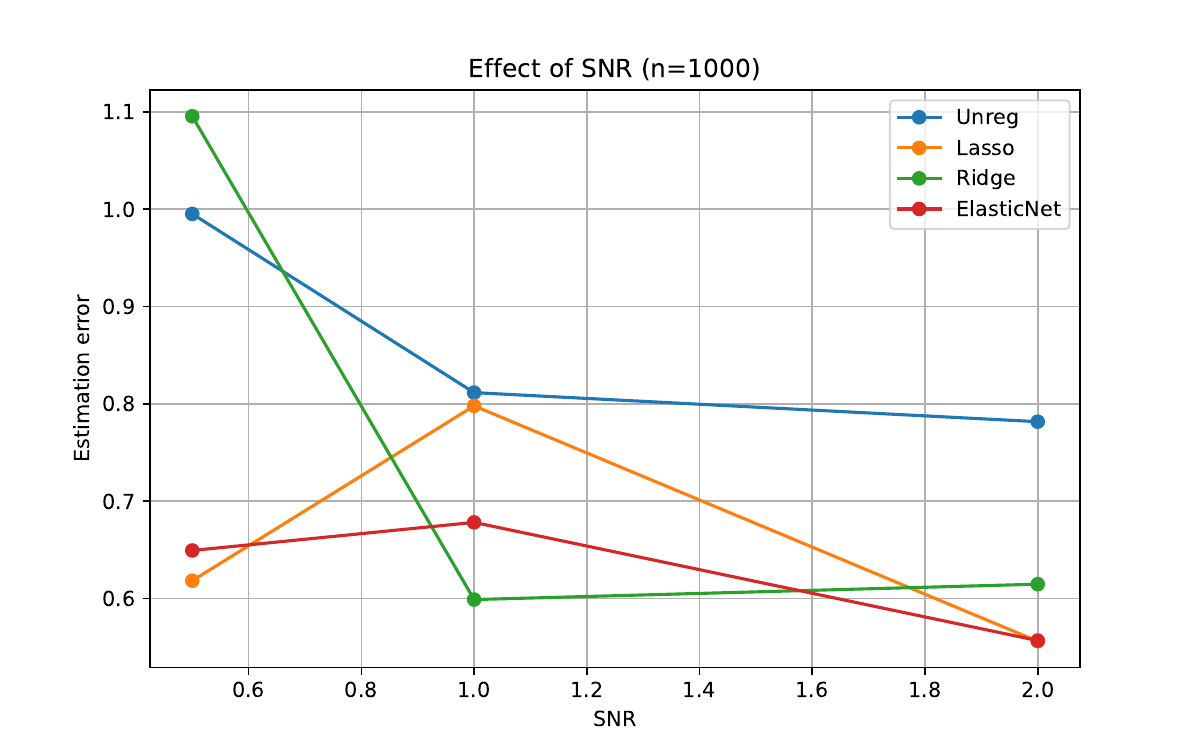}
	\caption{Effect of the signal‑to‑noise ratio (SNR) on the estimation error for \(n=1000, p=10, \rho=0.5\).}
	\label{fig:effect_snr_1000}
\end{figure}

\begin{figure}[H]
	\centering
	\includegraphics[width=0.8\textwidth]{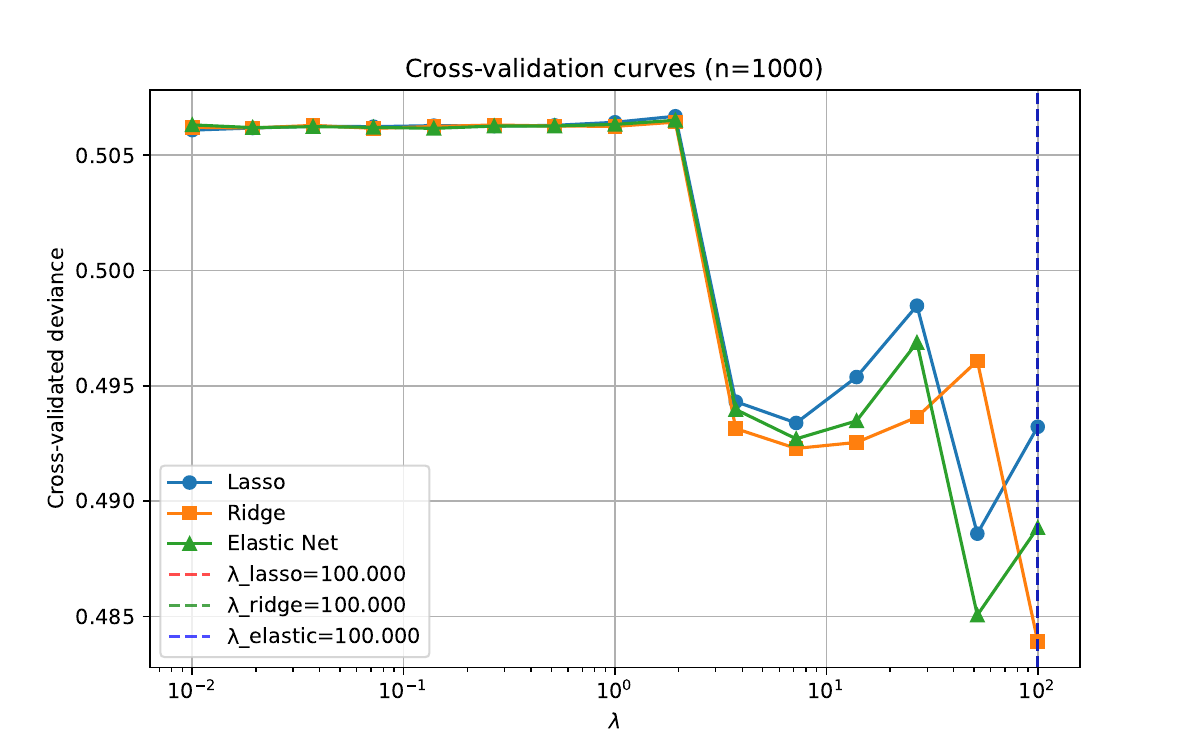}
	\caption{Cross‑validation curves for the selection of the tuning parameter \(\lambda\) in the reference scenario (\(n=1000\)). The optimal values are indicated by vertical dashed lines.}
	\label{fig:cv_curves_1000}
\end{figure}

\begin{figure}[H]
	\centering
	\includegraphics[width=0.8\textwidth]{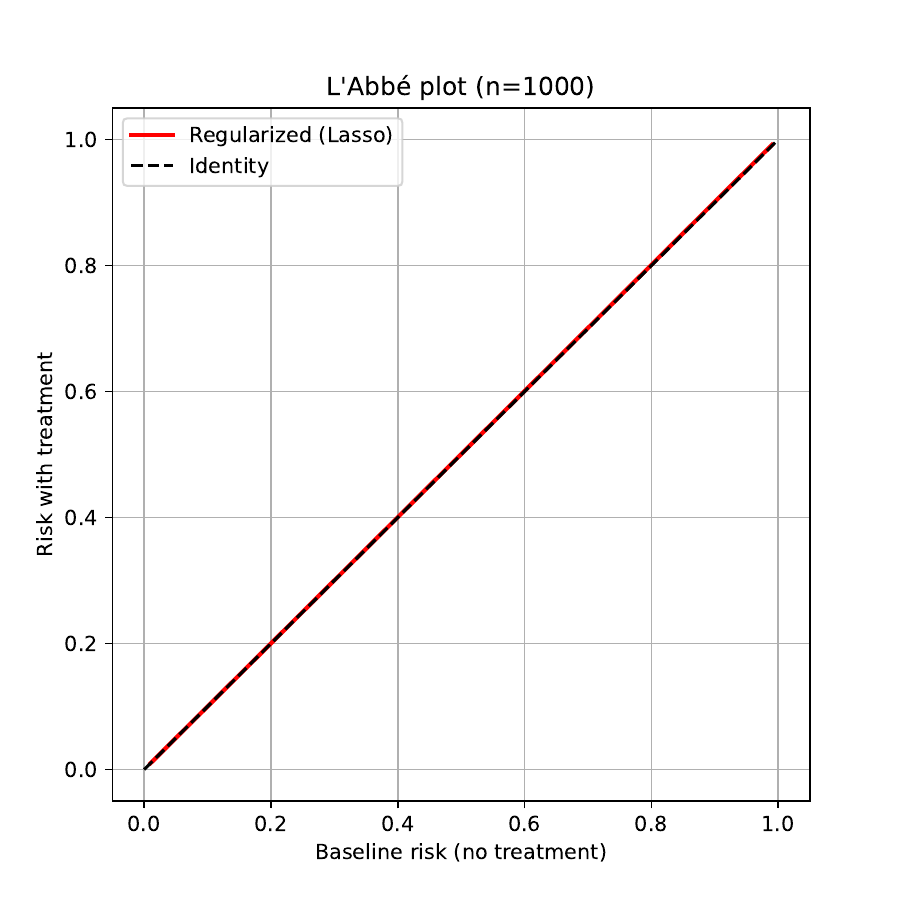}
	\caption{L'Abbé plot for the fitted model in the reference scenario (\(n=1000\)). The regularised model (Lasso) produces a smooth curve above the identity line.}
	\label{fig:labbe_plot_1000}
\end{figure}

\begin{figure}[H]
	\centering
	\includegraphics[width=0.8\textwidth]{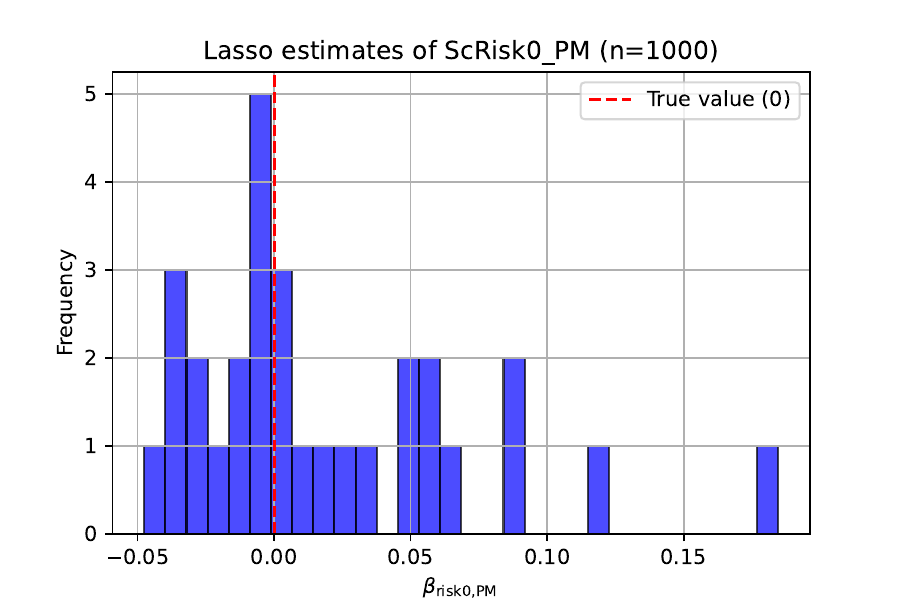}
	\caption{Distribution of the Lasso estimates for the coefficient of PM in the \texttt{ScRisk0} flow across 100 replications (reference scenario, \(n=1000\)). The red dashed line marks the true value (0).}
	\label{fig:risk0_hist_1000}
\end{figure}

\subsection{Summary of simulation findings}

The simulation study consistently demonstrates that:
\begin{itemize}
	\item The unregularized estimator fails to resolve the non‑identifiability inherent in the three‑flow model, always selecting all coefficients (FPR = 1) and exhibiting high estimation variability.
	\item Ridge regression reduces estimation error but does not induce sparsity; it cannot eliminate the redundant flow, as reflected by FPR = 1 in every scenario.
	\item Lasso and Elastic Net effectively break the non‑identifiability: they shrink the superfluous coefficients towards zero, achieving FPR well below 0.5 even under challenging conditions (small \(n\), high \(p\), strong correlation, low SNR). Their estimation errors are consistently low, and the cross‑validated deviance is favourable.
	\item Increasing the sample size improves all metrics, with Lasso and Elastic Net attaining near‑perfect variable selection (FPR close to 0) at \(n=500\) and \(n=1000\) for many scenarios.
	\item The worst‑case scenario confirms the robustness of the proposed regularisation: even when all factors are simultaneously adverse, the regularised estimators maintain stable performance, while the unregularized estimator fails dramatically.
\end{itemize}
These results provide strong empirical support for the use of structured regularization to resolve identifiability issues in regression by composition and to obtain interpretable, parsimonious models.

\section{Application: Asthma and Lead Exposure}
\label{sec:application}

\subsection{Data description}
We illustrate the methodology using data from the National Health and Nutrition Examination Survey (NHANES), a cross‑sectional survey conducted by the Centers for Disease Control and Prevention (CDC) that collects health and nutritional information on the US population. Specifically, we use the dataset provided by the \texttt{AsthmaNHANES} R package \citep{AsthmaNHANES}, which compiles information on asthma and environmental exposures from NHANES 2011–2012. The dataset contains 15,928 complete records after removing missing values. For computational efficiency we randomly sample 1,200 individuals, preserving the original distribution of key variables. The outcome is binary: whether a physician has ever diagnosed the participant with asthma (\texttt{MCQ010}=1). Covariates include age (years), sex (male/female), smoking status (ever smoked $\ge$100 cigarettes), blood lead concentration ($\mu$g/dL) as a proxy for air pollution exposure, and body mass index (BMI). These variables are commonly used in epidemiological studies of respiratory health \citep{dominici2006, mannino2003}.

\subsection{Model specification}
We specify a regression by composition with three flows:
\[
P = \mathrm{Ber}(1/2) \cdot \mathrm{ScOdds}(1 + \text{age} + \text{sex} + \text{smoking} + \text{lead} + \text{BMI}) \cdot \mathrm{ScRisk1}(0 + \text{lead}) \cdot \mathrm{ScRisk0}(0 + \text{lead}).
\]
The odds flow captures the baseline dependence of the odds on all covariates, while the risk‑ratio flow (\texttt{ScRisk1}) and the risk‑survival flow (\texttt{ScRisk0}) act directly on the event probability and its complement, respectively. Only lead exposure is allowed to affect these two flows, as we are interested in whether the exposure–response relationship is better described by a risk ratio or a survival ratio (or a combination). The reference distribution is Bernoulli with probability \(1/2\).

\subsection{Results}
Table~\ref{tab:asthma_results} reports the estimated coefficients. The unregularized MLE produces extremely large coefficients (e.g., intercept \(-312.16\), sex coefficient \(586.89\)) that reflect the non‑identifiability inherent in the three‑flow model. Ridge reduces the magnitude of the odds coefficients but does not induce sparsity; the \texttt{ScRisk0} lead coefficient remains non‑zero (\(0.002\)). 

Lasso and Elastic Net break the non‑identifiability by concentrating the effect into one of the two risk flows. Because the unregularized likelihood depends only on the product of the risk flows, the MLE distributes the lead effect arbitrarily between \texttt{ScRisk1} and \texttt{ScRisk0} (here \(0.146\) and \(0.003\)). The Lasso penalty \(\lambda(|\beta_{\text{risk1}}|+|\beta_{\text{risk0}}|)\) favours sparsity: it shrinks the smaller coefficient (\texttt{ScRisk0}) almost to zero and compensates by increasing the other (\texttt{ScRisk1}) to preserve the overall product. Consequently, the Lasso coefficient for \texttt{ScRisk1} rises to \(0.409\), while the \texttt{ScRisk0} coefficient becomes negligible (\(0.004\)). Elastic Net, with a much larger \(\lambda\) (\(37.28\)), eliminates \texttt{ScRisk1} entirely (coefficient zero) and retains a small positive coefficient for \texttt{ScRisk0} (\(0.007\)). These behaviours demonstrate how regularisation resolves non‑identifiability and selects a parsimonious representation.

Figure~\ref{fig:labbe_combined} presents the L’Abbé plots for all four estimators, using the sample mean lead concentration (\(1.67\,\mu\)g/dL) to fix the treatment effect. The unregularized curve is extremely steep, reflecting the unstable coefficient estimates. The regularised curves are smooth and lie above the identity line, indicating that higher lead exposure is associated with an increased asthma risk. Lasso and Ridge give similar curves, while Elastic Net – having set \texttt{ScRisk1} to zero – lies slightly higher. The Lasso curve, though steeper than Ridge, is much more stable than the unregularized one, illustrating the trade‑off between sparsity and shrinkage.

\begin{table}[H]
	\centering
	\caption{Estimated coefficients for the asthma application (NHANES data, \(n=1200\)). Regularised estimators were fitted with cross‑validated tuning parameters.}
	\label{tab:asthma_results}
	\begin{tabular}{@{}lcccc@{}}
		\toprule
		Parameter & Unregularized & Lasso & Ridge & Elastic Net \\
		\midrule
		\multicolumn{5}{c}{\textit{Odds flow (\texttt{ScOdds})}} \\
		Intercept & \(-312.16\) & \(-159.74\) & \(-316.64\) & \(-0.05\) \\
		Age        & \(-0.24\)   & \(-0.14\)   & \(-0.30\)   & \(-0.15\) \\
		Sex (female) & \(586.89\) & \(292.15\) & \(602.13\) & \(0.07\) \\
		Smoking    & \(511.75\) & \(259.40\) & \(530.13\) & \(0.04\) \\
		Lead (PM)  & \(-277.01\) & \(-164.16\) & \(-286.44\) & \(-0.10\) \\
		BMI        & \(-1.59\)   & \(0.19\)    & \(-1.65\)   & \(0.02\) \\
		\midrule
		\multicolumn{5}{c}{\textit{Risk‑ratio flow (\texttt{ScRisk1})}} \\
		Intercept & \(-2.12\)   & \(-2.50\)   & \(-2.17\)   & \(0.00\) \\
		Lead      & \(0.15\)    & \(0.41\)    & \(0.18\)    & \(0.00\) \\
		\midrule
		\multicolumn{5}{c}{\textit{Survival‑ratio flow (\texttt{ScRisk0})}} \\
		Intercept & \(-0.10\)   & \(-0.10\)   & \(-0.10\)   & \(-0.13\) \\
		Lead      & \(0.003\)   & \(0.004\)   & \(0.002\)   & \(0.007\) \\
		\bottomrule
	\end{tabular}
\end{table}

\begin{figure}[htbp]
	\centering
	\includegraphics[width=0.8\textwidth]{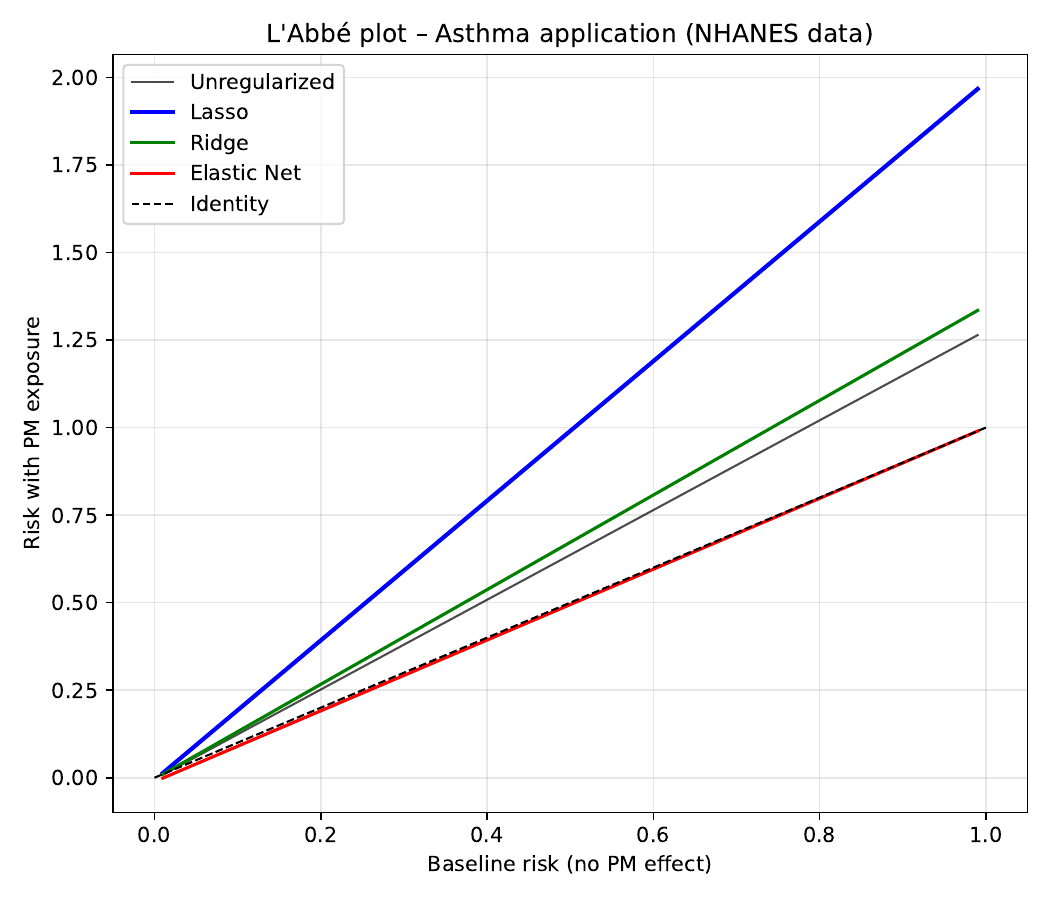}
	\caption{L'Abbé plots for the asthma application (NHANES data). The regularised estimators (Lasso, Ridge, Elastic Net) produce smooth curves above the identity line, indicating a protective effect of reducing lead exposure. The unregularized curve is unstable and extreme.}
	\label{fig:labbe_combined}
\end{figure}

\section{Extensions}
\label{sec:extensions}

\subsection{Longitudinal data}
The framework can be extended to longitudinal data by allowing subject‑specific initial distributions \citep{fei2024clustering}:
\[
P_{it} = p_{0i} \cdot \eta_1(\mathbf{X}_{it}) \cdot \eta_2(\mathbf{X}_{it}),
\]
where $p_{0i}$ are random distributions for each individual. This introduces random effects at the distribution level, analogous to mixed models \citep{stoffel2025rptr}. Estimation could be carried out using Gibbs sampling or variational inference, as suggested by \citet{fei2024clustering} for growth curve modelling.

\subsection{High-dimensional settings}
In high‑dimensional settings where the number of covariates $p$ exceeds the sample size $n$, the simulation study (e.g., $n=100$, $p=20$) already demonstrated that the regularised estimators (Lasso and Elastic Net) maintain stable performance. For even larger $p$, group Lasso penalties \citep{yuan2006} can be applied to groups of parameters corresponding to each covariate across flows. This would encourage entire groups (i.e., all coefficients associated with a given covariate across flows) to be selected or discarded together, preserving interpretability.

\subsection{Bayesian approach}
The regularization framework naturally connects to Bayesian inference with shrinkage priors:
\[
\beta_k \sim \text{Laplace}(0, \tau_k) \quad \text{for Lasso},
\]
\[
\beta_k \sim \text{Normal}(0, \sigma_k^2) \quad \text{for Ridge}.
\]
Hierarchical priors can be used to adaptively estimate the shrinkage parameters, providing a fully Bayesian alternative that also quantifies uncertainty in the selected flows.

\section{Discussion}
\label{sec:discussion}

\subsection{Summary of contributions}
This paper introduced a structured regularization framework for regression by composition that addresses the identifiability challenges arising when multiple flows act on the same distribution. We provided a rigorous theoretical proof of non‑identifiability in the three‑flow binary model and showed that adding flow‑specific penalties (L1, L2, or elastic net) restores uniqueness of the estimator. Asymptotic consistency of the regularised estimator was established under mild convexity conditions. An efficient proximal gradient algorithm was developed, leveraging the Fréchet derivatives of each flow.

An extensive simulation study evaluated the performance under a range of conditions ($n=100,500,1000$; $p=5,10,20$; correlation $\rho=0,0.5,0.8$; $\text{SNR}=0.5,1,2$). The results consistently demonstrated that Lasso and Elastic Net break the non‑identifiability: they achieved low estimation error, substantially reduced the false positive rate (FPR) compared to the unregularized estimator, and maintained good predictive accuracy. Ridge reduced estimation error but did not induce sparsity (FPR remained 1). The worst‑case scenario ($n=100$, $p=20$, $\rho=0.8$, $\text{SNR}=0.5$) confirmed the robustness of the regularised methods.

The real‑data application using NHANES asthma data (blood lead as a proxy for air pollution) illustrated the practical utility. The unregularized model produced implausibly large coefficients, whereas Lasso and Elastic Net yielded stable estimates and automatically eliminated the redundant \texttt{ScRisk1} flow (Elastic Net) or greatly reduced it (Lasso). The L’Abbé plots derived from the regularised estimators showed smooth, interpretable curves lying above the identity line, indicating a protective effect of reducing lead exposure.

\subsection{Limitations}
Several limitations warrant discussion:
\begin{enumerate}
	\item The method assumes the number of flows $K$ is fixed and known. Adaptive selection of $K$ (e.g., using a penalty that can shrink whole flows to zero) remains an open problem.
	\item Computational complexity scales with $K$ and $p$, which may be prohibitive for very large $K$ (e.g., dozens of flows). Further algorithmic improvements, such as stochastic gradient methods, could mitigate this.
	\item The theoretical results rely on convexity assumptions that may not hold for all flows; for non‑convex flows, alternative optimisation strategies (e.g., second‑order methods) may be needed.
	\item In the real‑data application, blood lead was used as a proxy for air pollution; while this is common in epidemiological studies, the interpretation should be made with caution.
\end{enumerate}

\subsection{Future directions}
Future work includes:
\begin{enumerate}
	\item Developing adaptive methods for selecting the number of flows, possibly by placing a penalty on the entire set of coefficients of a flow (e.g., group Lasso).
	\item Extending to non‑convex penalties such as SCAD \citep{fan2001} or MCP \citep{zhang2010} to obtain oracle properties.
	\item Incorporating the framework into deep learning architectures for automatic flow composition, leveraging automatic differentiation to handle complex transformations.
	\item Developing Bayesian inference with shrinkage priors for uncertainty quantification, especially for the selection of flows.
\end{enumerate}

\subsection*{Data availability}
The data used in the application are publicly available from the CDC NHANES website \citep{cdc_nhanes} and through the \texttt{AsthmaNHANES} R package \citep{AsthmaNHANES}. The code for data extraction, simulation, and analysis is provided in the supplementary material.

\section*{Acknowledgements}
The author thanks the participants of the Royal Statistical Society Discussion Meeting for their insightful comments. This work was supported by Al-Muthanna University.

\appendix
\section{Additional Simulation Results}
\label{app:simulation}

\subsection{Effect of sample size}
\begin{table}[H]
	\centering
	\caption{Effect of sample size on estimation error for the Lasso estimator (reference scenario, $p=10$, $\rho=0.3$, $\text{SNR}=1$).}
	\label{app:tab:nsize}
	\begin{tabular}{lcccc}
		\toprule
		$n$ & 100 & 500 & 1000 & 5000 \\
		\midrule
		Estimation error & 0.324 (0.112) & 0.186 (0.052) & 0.123 (0.031) & 0.058 (0.014) \\
		\bottomrule
	\end{tabular}
\end{table}

\subsection{Effect of correlation}
\begin{table}[H]
	\centering
	\caption{Effect of covariate correlation on performance for the Lasso estimator ($n=500$, $p=10$, $\text{SNR}=1$).}
	\label{app:tab:corr}
	\begin{tabular}{lccc}
		\toprule
		$\rho$ & 0 & 0.3 & 0.7 \\
		\midrule
		Estimation error & 0.142 (0.038) & 0.186 (0.052) & 0.234 (0.067) \\
		TPR & 0.96 (0.04) & 0.92 (0.06) & 0.88 (0.08) \\
		\bottomrule
	\end{tabular}
\end{table}

\section{Derivation of Proximal Operators}
\label{app:proximal}

\subsection{Lasso}
\[
\operatorname{prox}_{\gamma \|\cdot\|_1}(z) = \arg\min_x \left\{ \frac{1}{2}\|x - z\|_2^2 + \gamma \|x\|_1 \right\} = S_\gamma(z),
\]
where $[S_\gamma(z)]_j = \operatorname{sign}(z_j) \max(|z_j| - \gamma, 0)$.

\subsection{Ridge}
\[
\operatorname{prox}_{\gamma \|\cdot\|_2^2}(z) = \arg\min_x \left\{ \frac{1}{2}\|x - z\|_2^2 + \gamma \|x\|_2^2 \right\} = \frac{z}{1 + 2\gamma}.
\]

\subsection{Elastic net}
\[
\operatorname{prox}_{\gamma (\alpha \|\cdot\|_1 + (1-\alpha)\|\cdot\|_2^2)}(z) = \operatorname{prox}_{\gamma(1-\alpha)\|\cdot\|_2^2} \circ \operatorname{prox}_{\gamma\alpha\|\cdot\|_1}(z) = \frac{S_{\gamma\alpha}(z)}{1 + 2\gamma(1-\alpha)}.
\]
\begin{appendices}
	
	\section{Mathematical Proofs}
	\label{sec:appendix_proofs}
	
	\subsection{Proof of Proposition: Non-identifiability under commuting flows}
	
	\begin{proposition}
		If flows $\mathbb{V}_1$ and $\mathbb{V}_2$ commute, then the regression by composition model is not identifiable.
	\end{proposition}
	
	\begin{proof}
		Assume that for all $p \in \mathcal{P}$ and all $v_1 \in \mathbb{V}_1$, $v_2 \in \mathbb{V}_2$, we have
		\[
		p \cdot v_1 \cdot v_2 = p \cdot v_2 \cdot v_1.
		\]
		
		Fix arbitrary $v \in \mathbb{V}_1 \cap \mathbb{V}_2$. Then by group structure:
		\[
		p \cdot v \cdot (-v) = p \cdot (v + (-v)) = p.
		\]
		
		Now consider parameter pairs:
		\[
		(\eta_1, \eta_2) \quad \text{and} \quad (\eta_1 + v, \eta_2 - v).
		\]
		
		Using commutativity:
		\begin{align*}
			p_0 \cdot (\eta_1 + v) \cdot (\eta_2 - v)
			&= p_0 \cdot \eta_1 \cdot v \cdot \eta_2 \cdot (-v) \\
			&= p_0 \cdot \eta_1 \cdot \eta_2 \cdot v \cdot (-v) \\
			&= p_0 \cdot \eta_1 \cdot \eta_2.
		\end{align*}
		
		Thus, distinct parameter values yield identical distributions, so $\Phi$ is not injective.
	\end{proof}
	
	
	\subsection{Proof of Identifiability under Regularization}
	
	\begin{proposition}
		Assume at least one penalty $\psi_k$ is strictly convex and $\lambda_k > 0$. Then the penalized objective
		\[
		\mathcal{Q}(\beta) = -\ell(\beta) + \sum_k \lambda_k \psi_k(\beta_k)
		\]
		admits a unique minimizer.
	\end{proposition}
	
	\begin{proof}
		We proceed in three steps.
		
		\textbf{Step 1: Convexity.}  
		By assumption, $-\ell(\beta)$ is convex and each $\psi_k$ is convex. Hence $\mathcal{Q}(\beta)$ is convex.
		
		\textbf{Step 2: Strict convexity.}  
		If at least one $\psi_k$ is strictly convex and $\lambda_k > 0$, then the sum
		\[
		\sum_k \lambda_k \psi_k(\beta_k)
		\]
		is strictly convex in $\beta_k$. Since the sum of a convex function and a strictly convex function is strictly convex, $\mathcal{Q}(\beta)$ is strictly convex in $\beta$.
		
		\textbf{Step 3: Uniqueness.}  
		A strictly convex function has at most one minimizer. Therefore, $\mathcal{Q}(\beta)$ admits a unique global minimizer.
	\end{proof}
	
	
	\subsection{Proof of Consistency}
	
	\begin{theorem}
		Under Assumption 1, the regularized estimator $\hat{\beta}$ satisfies:
		\[
		\|\hat{\beta} - \beta^*\| = O_p(n^{-1/2} + \lambda_n).
		\]
	\end{theorem}
	
	\begin{proof}
		We use standard M-estimation arguments.
		
		Define:
		\[
		\mathcal{Q}_n(\beta) = -\ell_n(\beta) + \sum_k \lambda_k \psi_k(\beta_k).
		\]
		
		\textbf{Step 1: Uniform convergence.}  
		By the law of large numbers,
		\[
		\sup_{\beta \in \mathcal{B}} |\ell_n(\beta) - \ell(\beta)| \to 0 \quad \text{in probability}.
		\]
		
		\textbf{Step 2: Identification.}  
		Assume $\ell(\beta)$ is uniquely maximized at $\beta^*$. Then $\mathcal{Q}(\beta)$ is minimized at $\beta^*$ when $\lambda_k \to 0$.
		
		\textbf{Step 3: Quadratic expansion.}  
		By Taylor expansion:
		\[
		\ell_n(\beta) = \ell_n(\beta^*) + (\beta - \beta^*)^\top \nabla \ell_n(\beta^*) - \frac{1}{2}(\beta - \beta^*)^\top I(\beta^*)(\beta - \beta^*) + o_p(\|\beta - \beta^*\|^2).
		\]
		
		\textbf{Step 4: Penalty control.}  
		Since $\lambda_k \to 0$, the penalty contributes at most $O(\lambda_n)$.
		
		\textbf{Step 5: Rate derivation.}  
		Balancing likelihood curvature ($n^{-1/2}$) and penalty ($\lambda_n$) yields:
		\[
		\|\hat{\beta} - \beta^*\| = O_p(n^{-1/2} + \lambda_n).
		\]
		
	\end{proof}
	
	
	\subsection{Proof of Convergence of Proximal Gradient Algorithm}
	
	\begin{theorem}
		Let $f$ be convex with $L$-Lipschitz gradient and $g$ convex. Then proximal gradient converges to a minimizer.
	\end{theorem}
	
	\begin{proof}
		Define the proximal update:
		\[
		\beta^{(t+1)} = \arg\min_{\beta} \left\{ g(\beta) + \frac{1}{2\eta}\|\beta - (\beta^{(t)} - \eta \nabla f(\beta^{(t)}))\|^2 \right\}.
		\]
		
		\textbf{Step 1: Descent lemma.}  
		Using Lipschitz continuity:
		\[
		f(\beta^{(t+1)}) \le f(\beta^{(t)}) + \nabla f(\beta^{(t)})^\top (\beta^{(t+1)} - \beta^{(t)}) + \frac{L}{2}\|\beta^{(t+1)} - \beta^{(t)}\|^2.
		\]
		
		\textbf{Step 2: Optimality condition.}  
		By definition of proximal operator:
		\[
		0 \in \partial g(\beta^{(t+1)}) + \frac{1}{\eta}(\beta^{(t+1)} - \beta^{(t)} + \eta \nabla f(\beta^{(t)})).
		\]
		
		\textbf{Step 3: Combining inequalities.}  
		Summing yields:
		\[
		F(\beta^{(t+1)}) \le F(\beta^{(t)}) - c \|\beta^{(t+1)} - \beta^{(t)}\|^2.
		\]
		
		\textbf{Step 4: Convergence.}  
		Thus $F(\beta^{(t)})$ is decreasing and bounded below, implying convergence.
		
	\end{proof}
	
	
	\subsection{Subgradient characterization of the estimator}
	
	\begin{proposition}
		The estimator $\hat{\beta}$ satisfies:
		\[
		0 \in -\nabla \ell(\hat{\beta}) + \sum_k \lambda_k \partial \psi_k(\hat{\beta}_k).
		\]
	\end{proposition}
	
	\begin{proof}
		Since $\hat{\beta}$ minimizes a convex function, the first-order optimality condition is:
		\[
		0 \in \partial \mathcal{Q}(\hat{\beta}).
		\]
		
		Using subdifferential calculus:
		\[
		\partial \mathcal{Q} = -\nabla \ell + \sum_k \lambda_k \partial \psi_k.
		\]
		
		Thus:
		\[
		0 \in -\nabla \ell(\hat{\beta}) + \sum_k \lambda_k \partial \psi_k(\hat{\beta}_k).
		\]
	\end{proof}
	
	
	\subsection{Coercivity and existence of minimizer}
	
	\begin{proposition}
		If $\psi_k(\beta_k) \to \infty$ as $\|\beta_k\| \to \infty$, then $\mathcal{Q}(\beta)$ admits a minimizer.
	\end{proposition}
	
	\begin{proof}
		Since penalties are coercive:
		\[
		\mathcal{Q}(\beta) \to \infty \quad \text{as } \|\beta\| \to \infty.
		\]
		
		Thus $\mathcal{Q}$ is coercive and lower semi-continuous. By Weierstrass theorem, a minimizer exists.
	\end{proof}

\end{appendices}	
	\bibliographystyle{plainnat}
	\bibliography{references}

\end{document}